\documentclass[
  aip,
  amsmath,amssymb,
  preprint
]{revtex4-1}

\usepackage{graphicx}
\usepackage{amsmath,amssymb,amsthm}
\usepackage{booktabs}
\usepackage{array}
\usepackage{enumitem}
\usepackage{tikz}
\usetikzlibrary{patterns,calc,arrows.meta}
\usepackage{algpseudocode}
\usepackage{url} 

\usepackage{hyperref}

\newtheorem{theorem}{Theorem}[section]

\newtheorem{remark}{Remark}[section]

\begin{document}
\raggedbottom

\title{Electromagnetic Modes in Spherical Cavities: Complete Theory of Angular Spectra, Dispersion Relations, and Self-Adjoint Extensions}

\author{Mustafa Bakr}
\email{mustafa.bakr@physics.ox.ac.uk}
\affiliation{Clarendon Laboratory, Department of Physics, University of Oxford}
\author{Tongyu Zhang}
\author{Smain Amari}

\begin{abstract}
We present a complete theory of electromagnetic modes in spherical cavities, resolving fundamental questions about the nature of angular quantization. The standard result that angular indices $(\ell,m)$ must be integers is shown to be a consequence of domain constraints---regularity at both poles and single-valuedness in the azimuthal coordinate---rather than a requirement imposed by Maxwell's equations themselves. We prove that, for the sectoral case $\nu=m$, the function $\sin^{m}\theta$ exactly solves the angular eigenvalue equation for any real $m>0$, giving rise to a continuous dispersion curve. We demonstrate why non-sectoral modes (tesseral and zonal) appear only at isolated integer points on the full sphere, and show how boundary modifications such as cones and wedges convert these isolated points into continuous families of modes. Complete field solutions, wave impedances, and energy integrability conditions are derived. At the limiting point $(\nu, m) = (0, 0)$, the electromagnetic field vanishes identically while the underlying Debye potential remains non-trivial---a distinction with implications for mode counting that connects to longstanding questions in gauge theory and cavity quantization. Full-wave simulations validate the theoretical predictions with sub-percent accuracy. These results raise the possibility of structural analogues in wave equations on curved spacetimes, where conical deficits or horizon excisions similarly modify the angular domain.
\end{abstract}

\maketitle

\section{Introduction and Motivation}

The electromagnetic modes of a perfectly conducting spherical cavity constitute one of the foundational problems of classical electrodynamics. The complete solution, first systematically presented by Stratton~\cite{Stratton1941} and elaborated in the comprehensive treatise of Morse and Feshbach~\cite{MorseFeshbach1953}, expresses the resonant frequencies in terms of three discrete indices: the angular momentum index $\ell = 0, 1, 2, \ldots$, the azimuthal index $m$ satisfying $|m| \leq \ell$, and the radial index $n = 1, 2, 3,\ldots$. The resonant frequencies for transverse electric (TE) and transverse magnetic (TM) modes take the well-known forms
\begin{equation}
f_{\ell m n}^{\text{TE}} = \frac{c \cdot x_{\ell n}}{2\pi a}, \qquad
f_{\ell m n}^{\text{TM}} = \frac{c \cdot x'_{\ell n}}{2\pi a},
\label{eq:standard_resonances}
\end{equation}
where $a$ is the cavity radius, $c$ the speed of light, $x_{\ell n}$ denotes the $n$th zero of the spherical Bessel function $j_\ell(x)$, and $x'_{\ell n}$ denotes the $n$th zero of the derivative $[xj_\ell(x)]'$. These results appear in every modern textbook on electromagnetic theory, from Jackson~\cite{Jackson1999} to Pozar~\cite{Pozar2012}, and have served as the starting point for countless applications in microwave engineering, antenna design, and cavity quantum electrodynamics.

The standard explanation for the integer nature of $\ell$ proceeds as follows: separation of variables in the scalar Helmholtz equation leads to the associated Legendre equation for the polar angle dependence, and regularity of the solution at both poles $\theta = 0$ and $\theta = \pi$ requires the separation constant $\lambda = \nu(\nu+1)$ to take the form $\ell(\ell+1)$ for non-negative integer $\ell$. This argument, while mathematically correct for the full sphere with standard boundary conditions, obscures a deeper structural truth that forms the subject of the present investigation.

The situation may be illuminated by comparison with the cylindrical case. In a recent study of azimuthally propagating electromagnetic waves in cylindrical cavities~\cite{BakrAmari2023, BakrAmari2025}, we demonstrated that the azimuthal propagation constant $\nu$ varies continuously with frequency, with integer values appearing not as fundamental constraints but as resonance conditions imposed by the requirement of single-valued fields around the full azimuthal range. The dispersion relation $\nu(\omega)$ emerges as the primary physical quantity, and resonances represent discrete samples of this continuous curve. This perspective, which privileges propagation over resonance, proves effective for the analysis of complex microwave structures including dual-mode filters and multiplexers.

The natural question that motivates the present work is whether an analogous continuous structure exists for spherical cavities. Can the integer angular momentum indices be understood as samples from a continuous spectrum? If so, what is the nature of this underlying continuum, and how does it connect to the discrete resonances observed in practice? The answer, as we shall demonstrate, is both affirmative and nuanced. A continuous dispersion curve does exist, but its structure is richer than in the cylindrical case due to the two-dimensional nature of the spherical surface. Three distinct families of modes must be distinguished: sectoral modes with $\ell = |m|$, tesseral modes with $\ell > |m| > 0$, and zonal modes with $m = 0$. Of these, only the sectoral family admits a continuous extension to non-integer indices on the full sphere. The tesseral and zonal families exist only at isolated integer points, a restriction that can be traced to the singular behavior of the associated Legendre functions at the south pole for non-integer degree. This asymmetry between mode families, which appears not to have been explicitly recognized in the prior literature, has significant implications for both the mathematical structure of the problem and its physical applications.

The theoretical framework we develop draws upon the classical theory of singular Sturm-Liouville problems and their self-adjoint extensions, as exposited in the treatises of Coddington and Levinson~\cite{CoddingtonLevinson1955}, Reed and Simon~\cite{ReedSimon1975}, and Zettl~\cite{Zettl2005}. The angular Laplacian on the sphere possesses regular singular points at both poles, and the choice of boundary conditions at these singular endpoints---equivalently, the choice of self-adjoint extension of the formally symmetric differential operator---determines the spectral properties in a fundamental way. The full sphere with regularity demanded at both poles represents one particular self-adjoint realization, yielding the familiar discrete spectrum of integer angular momenta. Alternative boundary conditions, implemented physically through the introduction of conical or wedge-shaped conducting surfaces, select different self-adjoint extensions with qualitatively different spectral characteristics. This perspective reveals that the integer quantization of angular momentum in electromagnetic cavity problems is not an inevitable consequence of Maxwell's equations, but rather reflects the specific domain on which those equations are posed and the boundary conditions thereby imposed.

\subsection{The Spectral Landscape}

Before proceeding to detailed derivations, it is essential to carefully distinguish the two conceptually independent constraints that govern the electromagnetic spectrum of the spherical cavity. Conflating these distinct requirements---as is common in textbook treatments---obscures the physical origin of quantization and the possibilities for its relaxation.

\subsubsection{Azimuthal Propagation and the Quantization of $m$}
The separation of variables introduces an azimuthal factor $e^{im\phi}$, representing a wave propagating in the $\phi$-direction with wavenumber $m$. Locally, this propagation is entirely analogous to that in cylindrical geometries: a wave packet traveling around a circle of constant latitude has no knowledge of global topology, and its local propagation constant $m$ may assume any real value. 

The restriction to integer $m$ arises solely from the requirement of single-valuedness under the transformation $\phi \to \phi + 2\pi$. This is precisely the distinction between the Fourier transform, appropriate for infinite domains or local analysis with continuous spectral parameter, and the Fourier series, appropriate for periodic boundary conditions with discrete spectral parameter. The integer quantization of $m$ is a topological consequence of the $2\pi$-periodicity of the azimuthal coordinate, not a property of the differential equation itself.

When the azimuthal domain is restricted by conducting wedges to a range $0 < \phi < \Phi$ with $\Phi < 2\pi$, the single-valuedness constraint is replaced by boundary conditions at the wedge surfaces. The admissible values become $m = n\pi/\Phi$, which are generically non-integer, confirming that integer $m$ is indeed a consequence of domain geometry rather than an intrinsic electromagnetic requirement.

\subsubsection{Polar Regularity and the Constraint on $\nu$}
The polar dependence is governed by the associated Legendre equation
\begin{equation}
\frac{1}{\sin\theta}\frac{d}{d\theta}\left(\sin\theta\frac{d\Theta}{d\theta}\right) + \left[\nu(\nu+1) - \frac{m^2}{\sin^2\theta}\right]\Theta = 0,
\label{eq:angular_ode}
\end{equation}
which possesses regular singular points at the poles $\theta = 0$ and $\theta = \pi$. Unlike the azimuthal equation on a topologically closed circle, this is a boundary-value problem on the finite interval $(0, \pi)$ with singular endpoints. The constraint on $\nu$ arises from demanding that solutions remain regular at both poles---specifically, that the electromagnetic energy density be integrable near these points.

This polar boundary-value problem is of Sturm-Liouville type, and the admissible values of $\nu$ depend on the parameter $m$ appearing in the equation. The crucial observation, which forms the mathematical core of this paper, is that the nature of this dependence varies dramatically according to the relationship between $\nu$ and $m$.

\subsubsection{The Three Mode Families}

For integer indices satisfying the standard constraints, spherical harmonic modes fall into three families based on their nodal structure. The \textit{sectoral harmonics} have $\ell = |m|$, with angular dependence proportional to $(\sin\theta)^{|m|} e^{im\phi}$ concentrated near the equator. Their nodal lines are $2|m|$ meridional great circles dividing the sphere into sectors. The \textit{tesseral harmonics} satisfy $\ell > |m| > 0$, with nodal patterns combining latitudinal circles and meridional arcs to tessellate the sphere. The \textit{zonal harmonics} have $m = 0$ and are axially symmetric, with nodal lines at circles of constant latitude dividing the sphere into zones.

These families behave fundamentally differently under continuation to non-integer parameters:
For the \textit{sectoral case} $\nu = m$, we shall prove that the function $\Theta(\theta) = (\sin\theta)^m$ exactly solves equation~\eqref{eq:angular_ode} for any real $m > 0$, and remains regular at both poles. The condition $\nu = m$ thus defines a continuous curve in parameter space along which the polar boundary-value problem admits globally regular solutions, regardless of whether $m$ is an integer.
For \textit{non-sectoral cases}---where $\nu \neq m$---the situation differs qualitatively. The general solution regular at $\theta = 0$, expressible through associated Legendre functions $P_\nu^m(\cos\theta)$, develops a singularity at $\theta = \pi$ whenever $\nu$ is not an integer. For the zonal case $m = 0$, this singularity is logarithmic; for tesseral cases with $m \neq 0$, it is of power-law type. In both cases, the singular coefficient contains a factor $\sin(\nu\pi)$ that vanishes if and only if $\nu$ is an integer, explaining why non-sectoral modes exist only at isolated points in the $(\nu, m)$ plane.

\subsubsection{Independence of the Two Constraints $m$ and $\nu$}

The constraints on $m$ and $\nu$ are logically and physically independent. The integer nature of $m$ follows from azimuthal topology---the requirement of single-valuedness around a complete circuit in $\phi$. The constraint on $\nu$ follows from polar regularity---a boundary condition at the singular endpoints $\theta = 0$ and $\theta = \pi$ of the polar domain.

The standard textbook statement that ``$\ell$ and $m$ must be integers with $|m| \leq \ell$'' combines these distinct requirements without acknowledging their independence. A more precise statement is: (i) $m$ must be an integer for fields on the full azimuthal circle, but may be any real number if the azimuthal domain is restricted; (ii) for a given $m$, the degree $\nu$ must take integer values $\ell \geq |m|$ for regularity at both poles, \textit{except} along the sectoral line $\nu = m$ where regularity holds for all $m > 0$.

This distinction becomes physically consequential when boundary modifications are introduced. A wedge restricting the azimuthal range permits non-integer $m$ while leaving polar constraints unchanged. A cone removing the south pole relaxes polar regularity, permitting non-integer $\nu$ even for zonal modes. The combination of wedge and cone provides access to continuous two-parameter families that reveal the full structure underlying the discrete spectrum of the unmodified cavity.

\subsubsection{Comparison with the Cylindrical Case}

The contrast with cylindrical cavities illuminates why the spherical case is fundamentally richer. In a cylinder, the cross-sectional geometry involves a single angular coordinate $\phi$ with no singular points in its natural domain $[0, 2\pi)$. The radial Bessel equation and the azimuthal equation are independent, and both admit solutions for continuous parameter values. The azimuthal propagation constant $\nu$ varies continuously with frequency through the dispersion relation $\nu(\omega)$, with integer values appearing as resonance conditions when single-valuedness is imposed. Every mode family---regardless of radial structure---lies on such a continuous dispersion curve.

The sphere presents a qualitatively different situation. The two-dimensional angular domain $(0,\pi) \times [0, 2\pi)$ involves the polar coordinate $\theta$ whose equation possesses singular endpoints at both poles. The angular eigenvalue problem couples $\nu$ and $m$ through the associated Legendre equation, creating an interdependence absent in the cylindrical case. Global regularity at both poles collapses most of the $(\nu, m)$ parameter plane to an integer lattice, with only the sectoral line $\nu = m$ surviving as a continuous curve.

This explains why the ``azimuthal dispersion'' perspective, so natural and powerful for cylindrical structures, requires careful reformulation for spherical geometries. The continuous dispersion curve does exist---it is the sectoral line---but it represents only one family among three. The tesseral and zonal families have no cylindrical analogue: they are intrinsically discrete, existing only at isolated points with no underlying continuum.

Boundary modifications restore the cylinder-like behavior. A cone at $\theta = \theta_c$ truncates the polar domain, removing the south-pole singularity and converting the isolated zonal spectrum into a continuous family $\nu(\theta_c)$. A wedge of opening $\Phi$ samples the sectoral curve at $m = n\pi/\Phi$. The combination of cone and wedge recovers the full two-parameter continuous structure, analogous to the unrestricted parameter space of cylindrical modes.

\subsubsection{Summary of Contributions}

To position this work relative to the extensive prior literature on spherical harmonics and cavity modes, we enumerate the principal contributions that appear to be new:

First, we establish that the elementary function $\Theta(\theta) = (\sin\theta)^m$ exactly solves the associated Legendre equation for the sectoral case $\nu = m$ and \textit{any real} $m > 0$, not merely for positive integers. While this solution is well known for integer $m$ as the highest-weight spherical harmonic, its validity for continuous $m$ and its interpretation as a ``dispersion curve'' appear not to have been emphasized in the electromagnetic context.

Second, we provide a clear classification of the spectral structure: the sectoral family forms a continuous one-parameter curve in the $(\nu, m)$ plane, while the tesseral and zonal families exist only at isolated integer points when the full sphere is considered. This asymmetry, traceable to the $\sin(\nu\pi)$ factor in the singular coefficient of the Legendre function at the south pole, resolves the question of why spherical cavities do not exhibit the universal continuous dispersion found in cylindrical geometries.

Third, we frame the transition from discrete to continuous spectra through the theory of self-adjoint extensions: different physical boundary conditions (full sphere, cone-truncated, wedge-restricted) correspond to different self-adjoint realizations of the angular Laplacian, each with its characteristic spectral properties. This perspective unifies the treatment of various cavity geometries within a single mathematical framework.

Fourth, we validate the theoretical predictions through full-wave electromagnetic simulations. For wedge geometries, eight modes with non-integer azimuthal index $m=1/2$ show agreement within $1.4\%$ (mean $0.9\%$). For conical truncations accessing the continuous $\nu<1$ branch, six configurations spanning $\theta_c\in(0.4^\circ,34^\circ)$ yield sub-percent agreement with theory. These comparisons, performed without adjustable parameters, confirm the practical utility of the framework for cavity design.

\subsection{Organization of This Paper}

The remainder of this paper is organized as follows. Section~\ref{math} develops the mathematical foundations, beginning with the Frobenius analysis of the associated Legendre equation near its singular points and culminating in the proof that $(\sin\theta)^m$ exactly solves the equation when $\nu = m$. We demonstrate explicitly why this algebraic simplification occurs only on the sectoral line and why the tesseral and zonal cases necessarily involve singular functions for non-integer degree. Section~\ref{EM} presents the complete electromagnetic field solutions derived from the Debye potential formalism. All six field components are given explicitly for both TE and TM polarizations, and the wave impedances characterizing azimuthal power flow are derived. The dispersion relations connecting frequency to the angular indices are obtained, along with numerical data for a representative cavity. Section~\ref{boundary} addresses the role of boundary modifications in transforming the spectral structure. We show how conical conducting boundaries remove the south pole from the domain, eliminating the constraint that produces discrete zonal spectra and yielding instead a continuous family of modes parameterized by the cone angle. Similarly, azimuthal wedges sample the continuous sectoral dispersion curve at non-integer values determined by the wedge opening angle. The theoretical framework is that of self-adjoint extensions of the angular Laplacian, which provides a unified language for describing how different physical boundary conditions select different spectral realizations. Section~\ref{energy} treats energy considerations, demonstrating the conditions under which the field energy remains finite and connecting the limiting case of vanishing angular index to the theory of antenna feed structures and defect-supported modes. Section~\ref{validation} validates the theoretical framework through full-wave electromagnetic simulations, comparing predicted resonant frequencies against finite-element calculations for wedge geometries with non-integer azimuthal index and conical truncations accessing the continuous $\nu<1$ branch; agreement at the sub-percent level is achieved across all configurations tested. Section~\ref{antenna} establishes a direct structural connection between the cavity problem studied here and classical antenna theory, clarifying the relationship between bounded-domain cavity eigenmodes and radiating biconical antenna solutions through their shared angular operator and differing radial closure conditions. The paper concludes with a summary of results and a discussion of broader implications. Comprehensive appendices (\ref{app:connection_formulas}---\ref{app:zeros}) provide derivations of the second solution for the sectoral case, the singular behavior of Legendre functions at the antipodal point, asymptotic expansions of Bessel function zeros, and other technical material.

\section{Mathematical Foundations}
\label{math}
\subsection{The Separated Equations}

The scalar Helmholtz equation $(\nabla^2 + k^2)\Psi = 0$, which governs both Debye potentials for electromagnetic fields in source-free regions, takes the following form in spherical coordinates $(r, \theta, \phi)$:
\begin{equation}
\frac{1}{r^2}\frac{\partial}{\partial r}\left(r^2\frac{\partial\Psi}{\partial r}\right) + \frac{1}{r^2\sin\theta}\frac{\partial}{\partial\theta}\left(\sin\theta\frac{\partial\Psi}{\partial\theta}\right) + \frac{1}{r^2\sin^2\theta}\frac{\partial^2\Psi}{\partial\phi^2} + k^2\Psi = 0.
\label{eq:helmholtz_spherical}
\end{equation}
Here $k = \omega\sqrt{\varepsilon\mu}$ is the wavenumber, with $\omega$ the angular frequency and $\varepsilon$, $\mu$ the permittivity and permeability of the medium filling the cavity. The standard separation ansatz $\Psi(r,\theta,\phi) = R(r)\Theta(\theta)\Phi(\phi)$ leads to three ordinary differential equations. The azimuthal equation admits the elementary solution $\Phi(\phi) = e^{im\phi}$, where the separation constant $m$ must be an integer if single-valuedness is required over the full azimuthal range $0 \leq \phi < 2\pi$. For domains with restricted azimuthal extent, such as cavities containing conducting wedges, non-integer values of $m$ become permissible, a point to which we return in Part IV. The radial equation takes the form
\begin{equation}
r^2\frac{d^2R}{dr^2} + 2r\frac{dR}{dr} + [k^2r^2 - \nu(\nu+1)]R = 0,
\label{eq:radial_ode}
\end{equation}
which is recognized as the spherical Bessel equation of order $\nu$. The general solution of equation~\eqref{eq:radial_ode} is a linear combination of spherical Bessel functions:
\begin{equation}
R(r) = A j_\nu(kr) + B y_\nu(kr),
\label{eq:radial_general}
\end{equation}
where $j_\nu$ and $y_\nu$ are the spherical Bessel functions of the first and second kinds (the latter also called the spherical Neumann function). Their asymptotic behaviors near the origin are
\begin{equation}
j_\nu(x) \sim \frac{x^\nu}{(2\nu+1)!!}, \qquad y_\nu(x) \sim -\frac{(2\nu-1)!!}{x^{\nu+1}} \quad \text{as } x \to 0^+,
\label{eq:bessel_asymptotics}
\end{equation}
for $\nu > -1/2$, where $(2\nu+1)!! = (2\nu+1)(2\nu-1)\cdots$ denotes the double factorial. The two radial solutions have complementary convergence properties. Near the origin, $j_\nu(kr) \sim r^\nu$ while $y_\nu(kr) \sim r^{-(\nu+1)}$. The electromagnetic energy density involves field components that scale as derivatives of these functions; requiring finite total energy in a neighborhood of the origin yields the convergence conditions $\nu > -1/2$ for $j_\nu$ and $\nu < -1/2$ for $y_\nu$. Since the eigenvalue 
$\lambda = \nu(\nu+1)$ is invariant under $\nu \to -1-\nu$, each physical 
eigenvalue admits two representatives related by this symmetry. Choosing 
the representative with $\nu > -1/2$, the spherical Bessel function $j_\nu$ 
is always the finite-energy solution. This justifies both the restriction 
$\nu > -1/2$ and the choice $R(r) = A j_\nu(kr)$.

The connection between the radial and angular separation constants through equation~\eqref{eq:radial_ode} lies at the heart of the spectral theory developed herein. The angular equation, displayed previously as equation~\eqref{eq:angular_ode}, is the associated Legendre equation with degree $\nu$ and order $m$. This equation has been studied exhaustively since the nineteenth century, with comprehensive treatments given by Hobson~\cite{Hobson1931}, Whittaker and Watson~\cite{WhittakerWatson1927}, and in the modern era by Olver et al.~\cite{DLMF}. Nevertheless, certain structural features relevant to the electromagnetic problem appear not to have been fully appreciated, and we develop them in detail in what follows.

\subsection{Frobenius Analysis of the Angular Equation}
\label{subsec:frobenius}

The poles $\theta = 0$ and $\theta = \pi$, corresponding to $x = \cos\theta = +1$ and $x = -1$ respectively, are regular singular points of the associated Legendre equation. The local behavior of solutions near these points determines the global structure of the spectrum through what is known as the connection problem: how are the local solutions at one singular point expressed in terms of those at the other?

Consider first the north pole $\theta = 0$, corresponding to $x = 1$. Transforming equation~\eqref{eq:angular_ode} to the variable $x = \cos\theta$ yields the associated Legendre equation in its standard form:
\begin{equation}
(1-x^2)\frac{d^2y}{dx^2} - 2x\frac{dy}{dx} + \left[\nu(\nu+1) - \frac{m^2}{1-x^2}\right]y = 0.
\label{eq:legendre_assoc_x}
\end{equation}
Near $x = 1$, we introduce the local coordinate $\xi = 1 - x$ and seek solutions of the form $y = \xi^\alpha \sum_{n=0}^\infty a_n \xi^n$. The indicial equation, obtained by substituting this ansatz and collecting the leading-order terms, reads
\begin{equation}
\alpha(\alpha - 1) + \alpha - \frac{m^2}{4} = 0,
\end{equation}
from which the two roots $\alpha = \pm|m|/2$ follow immediately. Transforming back to the angular variable, the two independent local solutions near the north pole behave as
\begin{equation}
\Theta^{(+)}_{\text{north}}(\theta) \sim \theta^{|m|}, \qquad \Theta^{(-)}_{\text{north}}(\theta) \sim \theta^{-|m|}
\label{eq:north_pole_behavior}
\end{equation}
as $\theta \to 0^+$. For $m \neq 0$, the first solution vanishes at the pole while the second diverges; for $m = 0$, both solutions approach finite limits but with different derivatives. Physical solutions, requiring bounded energy density, must exclude the singular branch proportional to $\theta^{-|m|}$.

An entirely analogous analysis applies at the south pole $\theta = \pi$. With the local coordinate $\zeta = \pi - \theta$, the two independent solutions behave as
\begin{equation}
\Theta^{(+)}_{\text{south}}(\theta) \sim (\pi-\theta)^{|m|}, \qquad \Theta^{(-)}_{\text{south}}(\theta) \sim (\pi-\theta)^{-|m|}
\label{eq:south_pole_behavior}
\end{equation}
as $\theta \to \pi^-$. Again, regularity demands exclusion of the singular branch.

The crucial observation is that the local analyses at the two poles are independent. A solution that is regular at the north pole---proportional to $\theta^{|m|}$ as $\theta \to 0$---will generically be a linear combination of both regular and singular behaviors at the south pole. The requirement that the solution be simultaneously regular at both poles constitutes a global constraint that restricts the admissible values of $\nu$. It is this constraint that we now examine in detail, showing that its character depends fundamentally on whether $\nu$ equals $m$ (the sectoral case), $\nu$ exceeds $|m|$ (the tesseral case with $m \neq 0$), or $m$ vanishes (the zonal case).

\subsection{The Sectoral Case: Exact Solution for $\nu = m$}

We now establish the central mathematical result of this paper: when the degree $\nu$ equals the order $m$, the associated Legendre equation admits a simple closed-form solution that is regular at both poles for any real $m > 0$.

\begin{theorem}[Sectoral Solution]
\label{thm:sectoral}
For $\nu = m$ with $m$ any positive real number, the function
\begin{equation}
\Theta(\theta) = (\sin\theta)^m
\label{eq:sectoral_solution}
\end{equation}
is an exact solution of the associated Legendre equation~\eqref{eq:angular_ode}. This solution vanishes at both poles $\theta = 0$ and $\theta = \pi$ and is everywhere positive on the open interval $(0, \pi)$.
\end{theorem}

\begin{proof}
The proof proceeds by direct substitution. Writing $\Theta = \sin^m\theta$ for notational brevity, we compute the required derivatives:
\begin{align}
\frac{d\Theta}{d\theta} &= m\sin^{m-1}\theta\cos\theta, \\
\sin\theta\frac{d\Theta}{d\theta} &= m\sin^m\theta\cos\theta, \\
\frac{d}{d\theta}\left(\sin\theta\frac{d\Theta}{d\theta}\right) &= m\frac{d}{d\theta}(\sin^m\theta\cos\theta) \notag \\
&= m\left[m\sin^{m-1}\theta\cos^2\theta - \sin^{m+1}\theta\right].
\end{align}
The first term in equation~\eqref{eq:angular_ode} therefore becomes
\begin{equation}
\frac{1}{\sin\theta}\frac{d}{d\theta}\left(\sin\theta\frac{d\Theta}{d\theta}\right) = m^2\sin^{m-2}\theta\cos^2\theta - m\sin^m\theta.
\label{eq:first_term}
\end{equation}

For the potential term, we set $\nu(\nu+1) = m(m+1)$ and compute
\begin{align}
\left[\nu(\nu+1) - \frac{m^2}{\sin^2\theta}\right]\Theta &= \left[m(m+1) - \frac{m^2}{\sin^2\theta}\right]\sin^m\theta \notag \\
&= m(m+1)\sin^m\theta - m^2\sin^{m-2}\theta.
\label{eq:potential_term}
\end{align}

Adding equations~\eqref{eq:first_term} and~\eqref{eq:potential_term}:
\begin{align}
&m^2\sin^{m-2}\theta\cos^2\theta - m\sin^m\theta + m(m+1)\sin^m\theta - m^2\sin^{m-2}\theta \notag \\
&= m^2\sin^{m-2}\theta(\cos^2\theta - 1) + m^2\sin^m\theta \notag \\
&= -m^2\sin^{m-2}\theta \cdot \sin^2\theta + m^2\sin^m\theta \notag \\
&= -m^2\sin^m\theta + m^2\sin^m\theta = 0.
\end{align}
The function $(\sin\theta)^m$ thus satisfies equation~\eqref{eq:angular_ode} identically when $\nu = m$, completing the proof.
\end{proof}

Several remarks clarify this result. First, the solution is valid for any positive real $m$, not merely for positive integers. When $m$ is a positive integer, $(\sin\theta)^m$ coincides (up to normalization) with the standard associated Legendre function $P_m^m(\cos\theta)$, but the closed form extends unchanged to non-integer values. This extension provides the mathematical basis for the continuous sectoral branch mentioned in the introduction. Second, the algebraic cancellation demonstrated in the proof is specific to the case $\nu = m$. For $\nu \neq m$, no such simplification occurs, and one must work with the full machinery of hypergeometric functions to express solutions. We shall see below that this general case necessarily involves singular behavior at one or both poles for non-integer $\nu$. Third, the boundary behavior of the sectoral solution is manifest: as $\theta \to 0$, we have $(\sin\theta)^m \sim \theta^m \to 0$, and similarly as $\theta \to \pi$, $(\sin\theta)^m \sim (\pi-\theta)^m \to 0$. The solution vanishes at both poles with precisely the regular Frobenius exponent $+|m|$, as required for finite energy density. For non-integer $m$, while the function remains real and positive on $(0,\pi)$, the azimuthal factor $e^{im\phi}$ becomes multi-valued, necessitating a physical domain that excludes the full azimuthal range---a point developed in Part IV in connection with wedge geometries.

The second linearly independent solution for the case $\nu = m$ can be constructed by the method of reduction of order. If $\Theta_1 = \sin^m\theta$ is one solution, the second takes the form
\begin{equation}
\Theta_2(\theta) = \sin^m\theta \int^\theta \frac{d\theta'}{\sin\theta' \cdot \sin^{2m}\theta'} = \sin^m\theta \int^\theta \frac{d\theta'}{\sin^{2m+1}\theta'}.
\label{eq:second_solution_sectoral}
\end{equation}
Near $\theta = 0$, the integrand behaves as $\theta'^{-(2m+1)}$, yielding $\Theta_2 \sim \theta^m \cdot \theta^{-2m} = \theta^{-m}$ after multiplication by the prefactor $\sin^m\theta \sim \theta^m$. This is precisely the singular Frobenius branch, diverging as $\theta^{-m}$ at the north pole. By a similar argument, $\Theta_2$ diverges at the south pole as well. The solution $(\sin\theta)^m$ is therefore the unique globally regular solution of the associated Legendre equation for $\nu = m$, establishing its distinguished role in the spectral theory.

\subsection{Non-Sectoral Cases: The Origin of Discreteness}

We now demonstrate why the tesseral and zonal mode families do not admit continuous extensions to non-integer degree on the full sphere. The key lies not in any local asymmetry between the poles---both $\theta = 0$ and $\theta = \pi$ are regular singular points with identical Frobenius exponents $\pm|m|$, as established in Section~\ref{subsec:frobenius}---but rather in the \emph{global connection problem} between them.

The associated Legendre function of the first kind, $P_\nu^m(\cos\theta)$, is defined as the solution that is regular at the north pole $\theta = 0$. This is a choice of normalization: one selects the $\theta^{+|m|}$ Frobenius branch at $\theta = 0$ and discards the singular $\theta^{-|m|}$ branch. The nontrivial question is whether this globally defined function remains regular at the south pole $\theta = \pi$, or whether the analytic continuation from north to south pole introduces a singular component.

For the connection problem between the two poles, we must examine how solutions regular at the north pole behave at the south pole. This analysis is complicated by the various conventions for associated Legendre functions in the literature.

The Ferrers function of the first kind, as defined in the DLMF~\cite{DLMF} (\S14.3.1), takes the form
\begin{equation}
\mathsf{P}_\nu^\mu(x) = \left(\frac{1+x}{1-x}\right)^{\mu/2} \mathbf{F}\left(\nu+1, -\nu; 1-\mu; \frac{1-x}{2}\right),
\label{eq:ferrers_DLMF}
\end{equation}
where $\mathbf{F}(a,b;c;z) = {}_2F_1(a,b;c;z)/\Gamma(c)$ is Olver's regularized hypergeometric function. For $\mu > 0$ and non-integer, the prefactor diverges as $x \to 1^-$ (north pole), so this representation does \emph{not} directly give the solution regular at $\theta = 0$.

For non-negative \emph{integer} $\mu = m$, the representation~\eqref{eq:ferrers_DLMF} is indeterminate (the prefactor diverges while $1/\Gamma(1-m) = 0$ for $m \geq 1$), and the proper limiting form is given by DLMF 14.3.4:
\begin{equation}
\mathsf{P}_\nu^m(x) = \frac{(-1)^m \Gamma(\nu+m+1)}{2^m \Gamma(\nu-m+1)} (1-x^2)^{m/2} \mathbf{F}\left(\nu+m+1, m-\nu; m+1; \frac{1-x}{2}\right).
\label{eq:ferrers_integer_m}
\end{equation}
The crucial factor $(1-x^2)^{m/2} = \sin^m\theta$ ensures that this function \emph{vanishes} at both poles for $m \geq 1$, consistent with the regular Frobenius branch $\theta^{+|m|}$ identified in Section~\ref{subsec:frobenius}.

For the physical problem with non-integer $m > 0$ (arising from wedge geometries), we must construct the solution regular at the north pole directly from Frobenius analysis, rather than using the standard Ferrers function. Near $\theta = 0$, the two independent solutions behave as $\theta^{+m}$ (regular) and $\theta^{-m}$ (singular); we select the regular branch. This solution, when analytically continued to the south pole, generically acquires a singular component---and it is this connection problem that restricts $\nu$ to specific values.

The key question is therefore not the local regularity at either pole individually---which can always be achieved by selecting the appropriate Frobenius branch---but whether a \emph{single} solution can be simultaneously regular at \emph{both} poles. For the sectoral case $\nu = m$, the explicit solution $\Theta(\theta) = \sin^m\theta$ manifestly has this property for any $m > 0$, integer or not. For non-sectoral cases, the connection formulas are derived in Appendix~\ref{app:connection_formulas}.

The theory of linear transformations for hypergeometric functions, developed classically by Riemann and Kummer and exposited in modern form in the Digital Library of Mathematical Functions~\cite{DLMF}, provides the necessary connection formulas. The result, specialized to the associated Legendre case, is that $P_\nu^m(\cos\theta)$ contains a term proportional to $(1+\cos\theta)^{-m/2}$ multiplied by a function that approaches a non-zero limit as $\theta \to \pi$. Since $1 + \cos\theta = 2\cos^2(\theta/2) \sim (\pi-\theta)^2/2$ as $\theta \to \pi$, this term behaves as $(\pi-\theta)^{-m}$---the singular Frobenius branch.

For the zonal case $m = 0$, the singularity takes a logarithmic form. The Legendre function $P_\nu(\cos\theta) = P_\nu^0(\cos\theta)$ satisfies the connection formula
\begin{equation}
P_\nu(\cos\theta) = \frac{\sin(\nu\pi)}{\pi}\left[\psi(\nu+1) - \ln\left(\frac{1-\cos\theta}{2}\right)\right] + \text{regular terms},
\label{eq:Pnu_south_pole}
\end{equation}
where $\psi$ denotes the digamma function. The coefficient $\sin(\nu\pi)$ vanishes if and only if $\nu$ is an integer, in which case the logarithmic singularity is absent and $P_\nu$ reduces to a Legendre polynomial, regular at both poles.

This integrality condition has a geometric interpretation: the factor $\sin(\nu\pi)$ measures the failure of the solution to close smoothly upon traversing the full angular domain from pole to pole. For non-integer $\nu$, the analytic continuation of $P_\nu^m$ around the sphere acquires a singular admixture---a monodromy effect that is eliminated precisely when $\nu$ is an integer.

These results establish the following fundamental dichotomy:

\begin{theorem}[Spectral Structure on the Full Sphere]
\label{thm:spectral_structure}
Consider the associated Legendre equation~\eqref{eq:angular_ode} on the interval $\theta \in (0, \pi)$, with the requirement of regularity at both endpoints.

\noindent (i) On the sectoral line $\nu = m$ with $m > 0$, a regular solution exists for every positive real $m$, given explicitly by $(\sin\theta)^m$.

\noindent (ii) For the zonal case $m = 0$, regular solutions exist only when $\nu$ is a non-negative integer: $\nu = 0, 1, 2, \ldots$

\noindent (iii) For tesseral parameters $|m| < \nu$ with $m \neq 0$, regular solutions exist only when $\nu$ is a non-negative integer at least as large as $|m|$: $\nu = |m|, |m|+1, |m|+2, \ldots$
\end{theorem}

The proof of parts (ii) and (iii) follows from the singularity analysis above: the coefficient of the singular term at $\theta = \pi$ vanishes precisely for integer $\nu$, due to the factor $\sin(\nu\pi)$ in the connection formulas. Full details are provided in Appendix~\ref{app:connection_formulas}.

Theorem~\ref{thm:spectral_structure} has direct physical consequences. The sectoral family, alone among the three mode types, admits a continuous interpolation between its integer members. This continuity reflects the special algebraic structure of the $\nu = m$ case and has no analogue for the other families. The tesseral and zonal modes are intrinsically discrete: they exist only at isolated points in the $(\nu, m)$ parameter space, with no continuous curves connecting them on the full sphere.

\subsection{Graphical Representation of the Spectral Landscape}

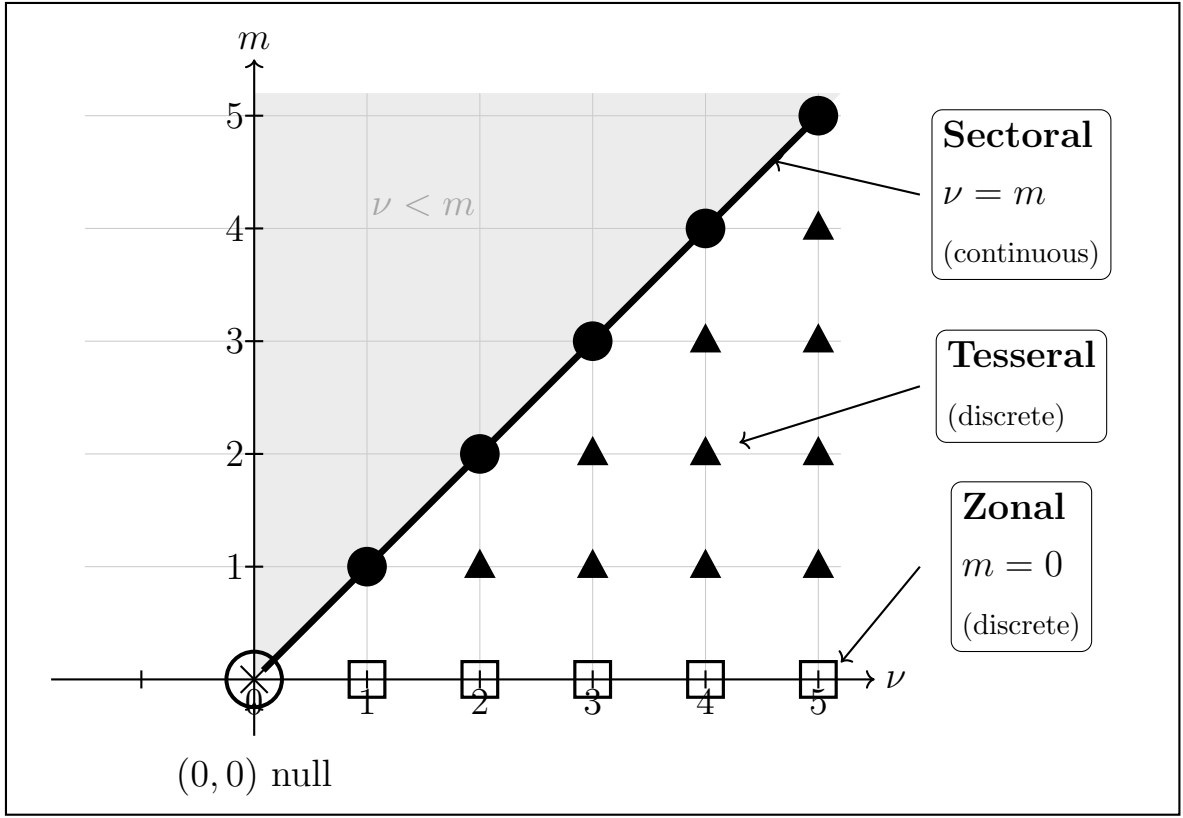
\begin{figure}[t]
\centering
\begin{tikzpicture}[scale=1.5]
    \draw[thick] (-2.2,-1.2) rectangle (8.2,6);
    
    \fill[gray!15] (0,0) -- (0,5.2) -- (5.2,5.2) -- cycle;
    
    \draw[->,thick] (-1.8,0) -- (5.5,0) node[right] {\large $\nu$};
    \draw[->,thick] (0,-0.3) -- (0,5.5) node[above] {\large $m$};
    
    \foreach \x in {1,2,3,4,5} {
        \draw[gray!40,thin] (\x,0) -- (\x,5.2);
    }
    \foreach \y in {1,2,3,4,5} {
        \draw[gray!40,thin] (-1.5,\y) -- (5.2,\y);
    }
    
    \node[below] at (0,0) {\large $0$};
    \foreach \x in {1,2,3,4,5} {
        \node[below] at (\x,0) {\large $\x$};
    }
    \foreach \y in {1,2,3,4,5} {
        \node[left] at (0,\y) {\large $\y$};
    }
    
    \foreach \x in {-1,1,2,3,4,5} {
        \draw[thick] (\x,0.08) -- (\x,-0.08);
    }
    \foreach \y in {1,2,3,4,5} {
        \draw[thick] (0.08,\y) -- (-0.08,\y);
    }
    
    \node[gray!70] at (1.5,4.2) {\large $\nu < m$};
    
    \draw[black, line width=2.5pt] (0.08,0.08) -- (5.1,5.1);
    
    \foreach \n in {1,2,3,4,5} {
        \fill[black] (\n,\n) circle (5pt);
    }
    
    \foreach \n in {1,2,3,4,5} {
        \draw[black, line width=1.2pt] (\n,0) +(-4.5pt,-4.5pt) rectangle +(4.5pt,4.5pt);
    }
    
    \foreach \x/\y in {2/1, 3/1, 3/2, 4/1, 4/2, 4/3, 5/1, 5/2, 5/3, 5/4} {
        \fill[black] (\x,\y) +(-0.14,-0.1) -- +(0.14,-0.1) -- +(0,0.16) -- cycle;
    }
    
    \draw[black, line width=1.5pt] (0,0) circle (7pt);
    \draw[black, thick] (-0.12,-0.12) -- (0.12,0.12);
    \draw[black, thick] (-0.12,0.12) -- (0.12,-0.12);
    
    \node[align=left, fill=white, draw=black, rounded corners, inner sep=4pt] 
        at (6.8,4.3) {\large\textbf{Sectoral}\\\large $\nu = m$\\(continuous)};
    \draw[black, thick, ->] (5.9,4.3) -- (4.6,4.6);
    
    \node[align=left, fill=white, draw=black, rounded corners, inner sep=4pt] 
        at (6.8,1.0) {\large\textbf{Zonal}\\\large $m = 0$\\(discrete)};
    \draw[black, thick, ->] (5.9,1.0) -- (5.2,0.15);
    
    \node[align=left, fill=white, draw=black, rounded corners, inner sep=4pt] 
        at (6.8,2.6) {\large\textbf{Tesseral}\\(discrete)};
    \draw[black, thick, ->] (5.9,2.6) -- (4.3,2.1);
    
    \node[below] at (0,-0.6) {\large $(0,0)$ null};
    \draw[black, thick,->] (0,-0.5) -- (0,-0.2);
\end{tikzpicture}

\caption{Spectral landscape in the $(\nu, m)$ parameter plane. The sectoral line $\nu = m$ (solid line with filled circles) supports solutions for any real $m > 0$, forming a continuous branch. Zonal modes (open squares, $m = 0$) and tesseral modes (filled triangles, $\nu > m > 0$) exist only at integer points. The region $\nu < m$ (shaded) corresponds to parameters for which no square-integrable solution exists on the full sphere; the physical 
implications of this regime for restricted domains remain unexplored. The point $(0,0)$ marks the boundary of the sectoral family; here the electromagnetic field vanishes while the underlying potential remains non-trivial (see Section~\ref{sec:null}).}
\label{fig:spectral_landscape}
\end{figure}

The structure established above may be visualized in the $(\nu, m)$ parameter plane, as depicted schematically in Figure~\ref{fig:spectral_landscape}. The sectoral line $\nu = m$ forms a diagonal ray from the origin extending to infinity, on which solutions exist for all $m > 0$. Integer points $m = 1, 2, 3, \ldots$ on this line correspond to the standard sectoral modes of the full sphere, while non-integer points become accessible when the azimuthal domain is restricted by conducting wedges.

The zonal axis $m = 0$ is marked by isolated points at $\nu = 0, 1, 2, \ldots$, with no solutions existing between them. Similarly, vertical lines at each non-zero integer $m$ carry isolated tesseral points at $\nu = |m|+1, |m|+2, \ldots$. The point $(\nu, m) = (0, 0)$ is special: while the angular equation admits the regular solution $\Theta = 1$ and the full Debye potential $\Pi = j_0(kr)$ is well-defined, the electromagnetic field extracted from this potential vanishes identically, as we discuss in Section~\ref{sec:null}. This point represents the boundary of the sectoral line but does not itself correspond to a propagating electromagnetic mode. The conical truncations studied here access only the continuous branch $0<\nu<1$.

\section{Electromagnetic Field Solutions}

\label{EM}

\subsection{Debye Potential Formulation}

The electromagnetic field in a source-free region of a homogeneous, isotropic medium may be constructed from scalar Debye potentials~\cite{Debye1909}, an approach that proves particularly natural in spherical coordinates. This formulation, developed systematically by Debye~\cite{Debye1909} and exposited in the treatises of Stratton~\cite{Stratton1941} and Collin~\cite{Collin1991}, expresses the electric and magnetic fields in terms of two scalar functions $\Pi_e$ and $\Pi_m$, each satisfying the Helmholtz equation. The electric Debye potential $\Pi_e$ generates transverse magnetic (TM) modes in which the radial component of the magnetic field vanishes, while the magnetic Debye potential $\Pi_m$ generates transverse electric (TE) modes with vanishing radial electric field.

The field expressions take the form
\begin{align}
\mathrm{E} &= \nabla \times \nabla \times (\mathrm{r}\Pi_e) + i\omega\mu\nabla \times (\mathrm{r}\Pi_m), \label{eq:E_Debye} \\
\mathrm{H} &= \nabla \times \nabla \times (\mathrm{r}\Pi_m) - i\omega\varepsilon\nabla \times (\mathrm{r}\Pi_e), \label{eq:H_Debye}
\end{align}
where we employ the time-harmonic convention $e^{-i\omega t}$ throughout. The wavenumber $k = \omega\sqrt{\varepsilon\mu}$ and intrinsic impedance $\eta = \sqrt{\mu/\varepsilon}$ characterize the medium.

The combination $\mathbf{r}\Pi = \hat{r}(r\Pi)$ appearing in equations~\eqref{eq:E_Debye}--\eqref{eq:H_Debye} is essential: the factor of $r$ multiplying the scalar potential $\Pi$ ensures that the radial field components behave as $r^{\nu-1}$ rather than $r^{\nu-2}$ near the origin, a distinction crucial for energy integrability when $\nu < 1/2$. This convention, employed by Stratton~\cite{Stratton1941}, differs from formulations that define the Hertz vector as $\hat{r}\Pi$ without the additional factor of $r$.

For a separated solution of the form
\begin{equation}
\Pi(r, \theta, \phi) = A \cdot j_\nu(kr) \cdot \Theta_\nu^m(\theta) \cdot e^{im\phi},
\label{eq:separated_potential}
\end{equation}
the curl operations in equations~\eqref{eq:E_Debye} and~\eqref{eq:H_Debye} may be evaluated explicitly. The radial function $j_\nu(kr)$ is the spherical Bessel function of order $\nu$, regular at the origin, while $\Theta_\nu^m(\theta)$ denotes a solution of the associated Legendre equation with the appropriate regularity properties. For the sectoral case $\nu = m$, we have established that $\Theta_m^m(\theta) = \sin^m\theta$ provides the unique globally regular solution.

\subsection{Complete Field Components}

The six components of the electromagnetic field for TM modes, derived from the electric Debye potential alone, are as follows. Setting $\Pi_m = 0$ and $\Pi_e$ as in equation~\eqref{eq:separated_potential}:
\begin{align}
E_r &= \frac{\nu(\nu+1)}{r} A \, j_\nu(kr) \, \Theta_\nu^m(\theta) \, e^{im\phi}, \label{eq:Er_TM} \\[6pt]
E_\theta &= \frac{A}{r} \frac{d}{dr}[r j_\nu(kr)] \frac{d\Theta_\nu^m}{d\theta} e^{im\phi}, \label{eq:Etheta_TM} \\[6pt]
E_\phi &= \frac{imA}{r\sin\theta} \frac{d}{dr}[r j_\nu(kr)] \, \Theta_\nu^m(\theta) \, e^{im\phi}, \label{eq:Ephi_TM} \\[6pt]
H_r &= 0, \label{eq:Hr_TM} \\[6pt]
H_\theta &= -\frac{i\omega\varepsilon m A}{\sin\theta} j_\nu(kr) \, \Theta_\nu^m(\theta) \, e^{im\phi}, \label{eq:Htheta_TM} \\[6pt]
H_\phi &= -i\omega\varepsilon A \, j_\nu(kr) \frac{d\Theta_\nu^m}{d\theta} e^{im\phi}. \label{eq:Hphi_TM}
\end{align}

For TE modes, setting $\Pi_e = 0$ and $\Pi_m$ as in equation~\eqref{eq:separated_potential}:
\begin{align}
E_r &= 0, \label{eq:Er_TE} \\[6pt]
E_\theta &= \frac{i\omega\mu m A}{\sin\theta} j_\nu(kr) \, \Theta_\nu^m(\theta) \, e^{im\phi}, \label{eq:Etheta_TE} \\[6pt]
E_\phi &= i\omega\mu A \, j_\nu(kr) \frac{d\Theta_\nu^m}{d\theta} e^{im\phi}, \label{eq:Ephi_TE} \\[6pt]
H_r &= \frac{\nu(\nu+1)}{r} A \, j_\nu(kr) \, \Theta_\nu^m(\theta) \, e^{im\phi}, \label{eq:Hr_TE} \\[6pt]
H_\theta &= \frac{A}{r} \frac{d}{dr}[r j_\nu(kr)] \frac{d\Theta_\nu^m}{d\theta} e^{im\phi}, \label{eq:Htheta_TE} \\[6pt]
H_\phi &= \frac{imA}{r\sin\theta} \frac{d}{dr}[r j_\nu(kr)] \, \Theta_\nu^m(\theta) \, e^{im\phi}. \label{eq:Hphi_TE}
\end{align}

Several structural features of these expressions merit comment. The radial field components $E_r$ (for TM) and $H_r$ (for TE) contain the factor $\nu(\nu+1)$, which vanishes at $\nu = 0$. The $1/r$ dependence in the radial components, combined with the $r^\nu$ behavior of $j_\nu(kr)$ near the origin, yields fields that scale as $r^{\nu-1}$---integrable in the energy norm for all $\nu > -1/2$. The tangential components involve either $j_\nu(kr)$ or its radial derivative $[rj_\nu(kr)]'/r$, with the ratio of these functions determining the wave impedance. The azimuthal dependence enters through the factor $e^{im\phi}$ and, for the $\theta$-components, through the additional factor $m/\sin\theta$ that produces a divergence at the poles unless $m = 0$ or the angular function vanishes sufficiently rapidly there.

\subsection{The Null Field at $(\nu, m) = (0, 0)$}
\label{sec:null}
A question that arises naturally from the spectral structure established in Part~\ref{math} is why the point $(\nu, m) = (0, 0)$---which lies at the boundary of the sectoral line $\nu = m$ as $m \to 0^+$---does not correspond to a physical electromagnetic mode. The answer lies in the structure of the Debye potential representation and the operators that extract fields from potentials.

At $\nu = 0$ and $m = 0$, the angular eigenfunction is $\Theta_0^0(\theta) = P_0(\cos\theta) = 1$, a constant, and the azimuthal factor is $e^{i \cdot 0 \cdot \phi} = 1$, also constant. The separated potential reduces to
\begin{equation}
\Pi = A \cdot j_0(kr) = A \cdot \frac{\sin(kr)}{kr},
\end{equation}
depending only on the radial coordinate. Crucially, this potential is \emph{non-trivial}: it is a well-defined, regular solution to the scalar Helmholtz equation $(\nabla^2 + k^2)\Pi = 0$ throughout the cavity volume.

However, the electromagnetic fields generated by this potential vanish identically. For \textbf{TM modes} with $\nu = 0$, $m = 0$, we evaluate each field component from Eqs.~\eqref{eq:Er_TM}--\eqref{eq:Hphi_TM}:
\begin{align}
E_r &= \frac{\nu(\nu+1)}{r}\,A\,j_0(kr)\cdot 1 \cdot 1 = \frac{0}{r}\,A\,j_0(kr) = 0, \\[4pt]
E_\theta &= \frac{A}{r}\,[r j_0(kr)]'\,\frac{d(1)}{d\theta} = \frac{A}{r}\,[r j_0(kr)]'\cdot 0 = 0, \\[4pt]
E_\phi &= \frac{i\,0\,A}{r\sin\theta}\,[r j_0(kr)]'\cdot 1 = 0, \\[4pt]
H_\theta &= -\frac{i\omega\varepsilon\,0\,A}{\sin\theta}\,j_0(kr)\cdot 1 = 0, \\[4pt]
H_\phi &= -i\omega\varepsilon A\, j_0(kr)\,\frac{d(1)}{d\theta} = -i\omega\varepsilon A\, j_0(kr)\cdot 0 = 0.
\end{align}
Every component of the TM field vanishes identically. The potential $\Pi_e = A\,j_0(kr)$ exists but generates no electromagnetic field.

For \textbf{TE modes} with $\nu = 0$, $m = 0$, the situation is identical:
\begin{align}
E_\theta &= \frac{i\omega\mu \cdot 0 \cdot A}{\sin\theta} j_0(kr) \cdot 1 = 0, \\[4pt]
E_\phi &= i\omega\mu A \, j_0(kr) \cdot \frac{d(1)}{d\theta} = 0, \\[4pt]
H_r &= \frac{0 \cdot 1}{r} A \, j_0(kr) = 0, \\[4pt]
H_\theta &= \frac{A}{r} [r j_0(kr)]' \cdot 0 = 0, \\[4pt]
H_\phi &= \frac{i \cdot 0 \cdot A}{r\sin\theta} [r j_0(kr)]' \cdot 1 = 0.
\end{align}
Again, all field components vanish. 

The mathematical origin of this null result is clear from the structure of the Debye representation. The operator $\mathcal{L}: \Pi \mapsto \nabla \times \nabla \times (\mathbf{r}\Pi)$ that extracts electric fields from the Debye potential has a non-trivial \emph{kernel}: all functions $\Pi(r)$ that depend only on the radial coordinate. At $(0,0)$, the potential $\Pi = j_0(kr)$ lies precisely in this kernel. More explicitly, when both $\nu = 0$ and $m = 0$:
\begin{enumerate}
\item The factor $\nu(\nu+1) = 0$ eliminates the radial field components $E_r$ and $H_r$.
\item The azimuthal derivative $\partial/\partial\phi$ acting on $e^{im\phi}$ produces $im \cdot e^{im\phi} = 0$.
\item The polar derivative $d\Theta/d\theta$ acting on the constant function $\Theta = 1$ gives zero.
\end{enumerate}
No mechanism remains to generate nonzero fields from the scalar potential.

This result reflects physical reality: Maxwell's equations in free space admit no nontrivial spherically symmetric, time-harmonic electromagnetic field. A spherically symmetric electric field requires a point charge at the origin (electrostatics), while a spherically symmetric magnetic field is forbidden by $\nabla \cdot \mathbf{B} = 0$. The vanishing of the Debye-generated fields at $\nu = m = 0$ is thus physically necessary.

\begin{remark}[Potential versus field]
The distinction between the potential and the field at $(0,0)$ merits emphasis. The Debye potential $\Pi = j_0(kr)$ is a well-defined, non-singular solution to the scalar wave equation. It is the electromagnetic field $(\mathbf{E}, \mathbf{H})$ extracted from this potential that vanishes---a consequence of the differential structure of the curl operations, not of any pathology in the potential itself. This situation is analogous to pure gauge configurations in gauge field theory, where the vector potential $\mathbf{A}$ may be non-zero while the field strength $\mathbf{F} = d\mathbf{A}$ vanishes. The physical and mathematical interpretation of such configurations---particularly regarding mode counting in cavity quantization---has been the subject of ongoing discussion in the literature; see, e.g., Refs.~\cite{AharonovBohm1959, Berry1984}. We do not pursue this question further here.
\end{remark}


\subsection{Boundary Conditions and Dispersion Relations}

The boundary condition at a perfectly conducting spherical surface of radius $a$ requires the tangential components of the electric field to vanish. For TM modes, examination of equations~\eqref{eq:Etheta_TM} and~\eqref{eq:Ephi_TM} shows that both conditions reduce to the single requirement
\begin{equation}
\left.\frac{d}{dr}[r j_\nu(kr)]\right|_{r=a} = 0.
\label{eq:BC_TM}
\end{equation}
Denoting by $x'_{\nu n}$ the $n$th positive root of the equation $[xj_\nu(x)]' = 0$, the resonant wavenumbers are $k_{\nu n} = x'_{\nu n}/a$, yielding frequencies
\begin{equation}
f_{\nu n}^{\text{TM}} = \frac{c \, x'_{\nu n}}{2\pi a}.
\label{eq:freq_TM}
\end{equation}

For TE modes, equations~\eqref{eq:Etheta_TE} and~\eqref{eq:Ephi_TE} both require
\begin{equation}
j_\nu(ka) = 0.
\label{eq:BC_TE}
\end{equation}
With $x_{\nu n}$ denoting the $n$th positive root of $j_\nu(x) = 0$, the resonant frequencies are
\begin{equation}
f_{\nu n}^{\text{TE}} = \frac{c \, x_{\nu n}}{2\pi a}.
\label{eq:freq_TE}
\end{equation}

The dispersion relations~\eqref{eq:freq_TM} and~\eqref{eq:freq_TE} express the resonant frequency as a function of the angular index $\nu$ for each radial mode number $n$. When $\nu$ is restricted to non-negative integers, these reduce to the standard textbook formulas. The continuous extension to non-integer $\nu$---valid on the sectoral line, or with modified boundary conditions for other mode families---reveals the full dispersion structure underlying the discrete resonances.

\subsection{Asymptotic Dispersion Formulas}

For large angular index $\nu$, the zeros of spherical Bessel functions admit asymptotic expansions based on Airy function theory~\cite{Abramowitz1964, DLMF}. The McMahon expansion for the $n$th zero of $j_\nu(x)$ takes the form
\begin{equation}
x_{\nu,n} \approx \nu + a_n \nu^{1/3} + \frac{3a_n^2}{20}\nu^{-1/3} + O(\nu^{-1}),
\label{eq:bessel_zero_asymptotic}
\end{equation}
where $a_n$ is related to the $n$th zero of the Airy function $\mathrm{Ai}(-a_n) = 0$. For the first three zeros: $a_1 \approx 1.8558$, $a_2 \approx 3.2446$, $a_3 \approx 4.3817$.

For TM modes, the boundary condition involves the derivative $[xj_\nu(x)]' = 0$, leading to a different asymptotic expansion:
\begin{equation}
x'_{\nu,n} \approx \nu + a'_n \nu^{1/3} + O(\nu^{-1/3}),
\label{eq:bessel_deriv_zero_asymptotic}
\end{equation}
where $a'_n$ is related to zeros of the Airy function derivative $\mathrm{Ai}'(-a'_n) = 0$. For the first zero: $a'_1 \approx 0.8086$.

\subsubsection{TE Dispersion Formula}

For TE modes on the sectoral line $\nu = m$, the boundary condition $j_m(ka) = 0$ combined 
with equation~\eqref{eq:bessel_zero_asymptotic} yields
\begin{equation}
ka = x_{m,1} \approx m + 1.856 \, m^{1/3}.
\end{equation}
Inverting this relation gives the azimuthal index as a function of frequency:
\begin{equation}
\boxed{m_{\text{TE}}(\omega) \approx \frac{\omega a}{c} - 1.856 \left(\frac{\omega a}{c}\right)^{1/3}}
\label{eq:dispersion_m_omega_TE}
\end{equation}
or equivalently, expressing frequency as a function of mode index:
\begin{equation}
\omega_{\text{TE}}(m) \approx \frac{c}{a}\left(m + 1.856 \, m^{1/3}\right).
\label{eq:dispersion_omega_m_TE}
\end{equation}

\subsubsection{TM Dispersion Formula}

For TM modes, the boundary condition $[xj_m(x)]'|_{x=ka} = 0$ yields
\begin{equation}
ka = x'_{m,1} \approx m + 0.809 \, m^{1/3}.
\end{equation}
The inverted dispersion relation is:
\begin{equation}
\boxed{m_{\text{TM}}(\omega) \approx \frac{\omega a}{c} - 0.809 \left(\frac{\omega a}{c}\right)^{1/3}}
\label{eq:dispersion_m_omega_TM}
\end{equation}
or equivalently:
\begin{equation}
\omega_{\text{TM}}(m) \approx \frac{c}{a}\left(m + 0.809 \, m^{1/3}\right).
\label{eq:dispersion_omega_m_TM}
\end{equation}

\subsubsection{Physical Interpretation of the Asymptotic Coefficients}

These formulas are the spherical cavity analogues of whispering gallery mode dispersion 
relations. The dominant term $\omega a/c$ represents a ray-optics picture in which the 
wave circulates around the equator at the cavity surface. The negative correction 
proportional to $(\omega a/c)^{1/3}$ arises from the wave nature of the field: the 
transition between the oscillatory region near the cavity wall and the evanescent 
region toward the center is governed by Airy functions, whose characteristic length 
scale introduces the fractional power. At high frequencies (large $ka$), modes become 
increasingly confined to the equatorial region near $r = a$, and $m$ approaches $ka$ 
asymptotically.
The different coefficients for TE and TM modes (1.856 vs 0.809) have a precise 
mathematical origin: they derive from the zeros of different Airy functions. The TE 
boundary condition $j_\nu(ka) = 0$ connects asymptotically to zeros of $\mathrm{Ai}(-z)$, 
with the first zero at $z_1 \approx 2.338$, yielding the coefficient $a_1 \approx 1.856$ 
after the appropriate transformation. The TM boundary condition $[xj_\nu(x)]' = 0$ 
connects instead to zeros of the Airy derivative $\mathrm{Ai}'(-z)$, with the first zero 
at $z'_1 \approx 1.019$, yielding the smaller coefficient $a'_1 \approx 0.809$.
The physical consequence is that TM modes have lower frequencies than TE modes for the 
same angular index $m$: the TM dispersion curve lies below the TE curve in the 
$(m, \omega)$ plane. Equivalently, for a given frequency, the TM mode has a larger 
effective angular index.

Table~\ref{tab:dispersion_data} presents computed values of the universal roots. The table spans from $\nu=0$ to $\nu=3$, covering both integer and non-integer values accessible through boundary modifications. 
\begin{table}[h]
\centering
\caption{Universal roots and resonant frequencies for a spherical cavity of radius $a = 15$ mm. The roots $x_{\nu,1}$ (TE) and $x'_{\nu,1}$ (TM) depend only on $\nu$; frequencies scale as $f = cx/(2\pi a)$.}
\label{tab:dispersion_data}
\begin{tabular}{|c|c|c|c|c|}
\hline
$\nu$ & $x_{\nu,1}$ (TE) & $f_{\text{TE}}$ [GHz] & $x'_{\nu,1}$ (TM) & $f_{\text{TM}}$ [GHz] \\
\hline
$0$ & $\pi$ & 10.00 & $\pi/2$ & 5.00 \\
$0.50$ & 3.832 & 12.20 & 2.166 & 6.89 \\
$1$ & 4.493 & 14.30 & 2.744 & 8.73 \\
$1.50$ & 5.136 & 16.35 & 3.311 & 10.54 \\
$2$ & 5.763 & 18.35 & 3.870 & 12.32 \\
$2.50$ & 6.380 & 20.31 & 4.424 & 14.08 \\
$3$ & 6.988 & 22.24 & 4.973 & 15.83 \\
\hline
\end{tabular}
\end{table}
Several features of the data merit comment. At $\nu = 0$, the roots take the exact values $x_{0,1} = \pi$ (TE) and $x'_{0,1} = \pi/2$ (TM), corresponding to the $\ell = 0$ modes of a full sphere. As $\nu$ increases, both roots grow approximately linearly with a fractional-power correction as predicted by equations~\eqref{eq:dispersion_m_omega_TE} and~\eqref{eq:dispersion_m_omega_TM}. 

\begin{figure}[t]
\centering
   \includegraphics[width=\linewidth]{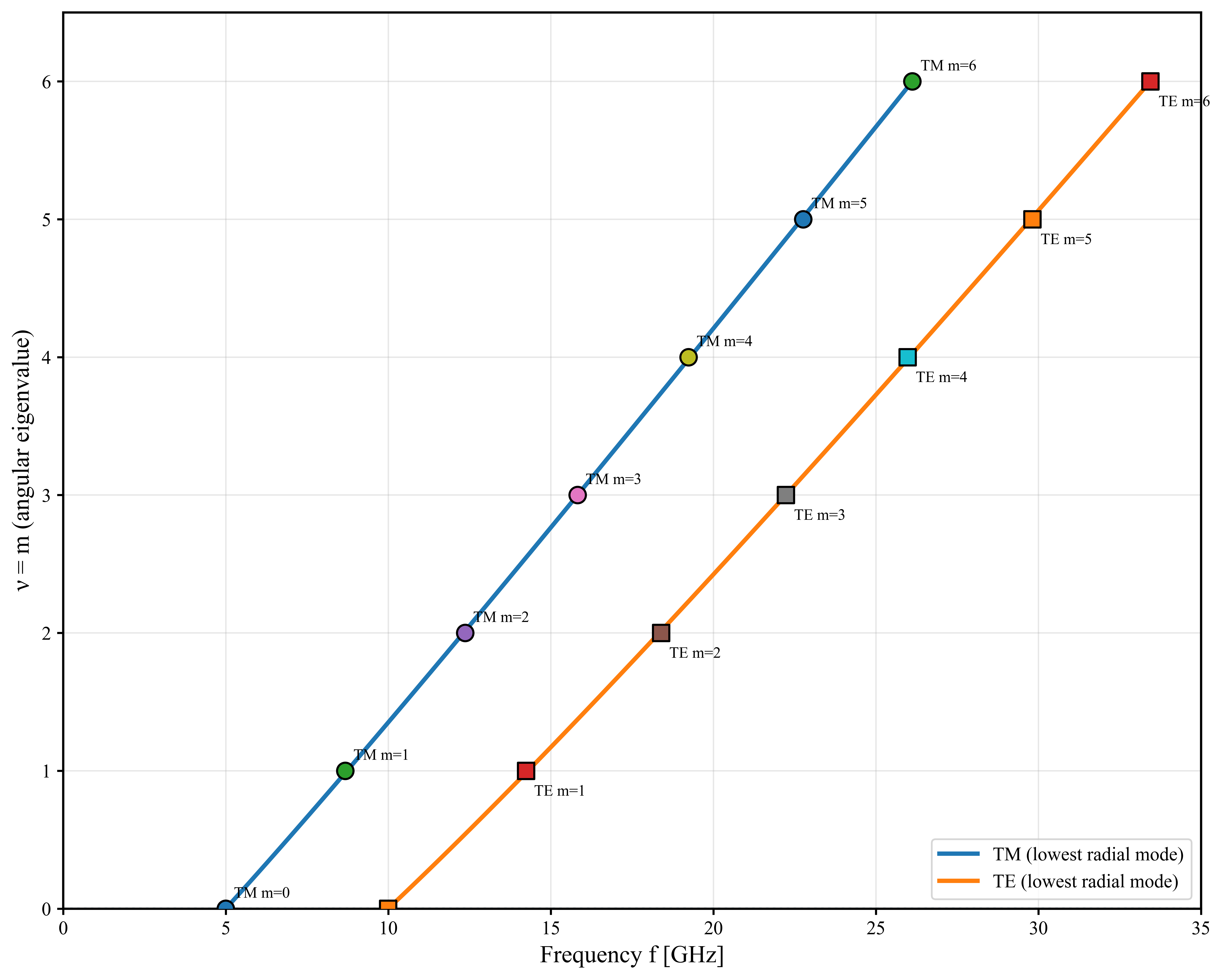}
\caption{Dispersion curves showing the first TE and TM roots as functions of $\nu$. The TE boundary condition $j_\nu(ka) = 0$ (blue) and TM boundary condition $[xj_\nu(x)]' = 0$ (orange) define continuous curves. Filled circles mark integer values corresponding to standard spherical cavity resonances.}
\label{fig:dispersion}
\end{figure}

For non-integer $\nu$ (realized through conical boundaries or on the sectoral branch with wedges), the resonant frequencies interpolate smoothly between the integer-$\nu$ values. This continuous variation validates the interpretation of $\nu$ as a continuous dispersion parameter, with integer values representing resonance conditions on the full sphere. Figure~\ref{fig:dispersion} presents the dispersion curves graphically, showing both TE and TM branches as continuous functions of the angular index $\nu$.

\subsection{Wave Impedances}

The concept of wave impedance, so fruitful in the analysis of waveguide propagation, extends naturally to azimuthally propagating modes in spherical cavities. For waves propagating in the $\phi$-direction, the relevant impedance is the ratio of appropriate transverse field components.

For TE modes, considering the field components transverse to the direction of azimuthal propagation, we define
\begin{equation}
Z_{\text{TE}} = \frac{E_\theta}{H_r} = {\frac{i\omega\mu m r}{\nu(\nu+1)\sin\theta}}.
\label{eq:Z_TE}
\end{equation}
This expression diverges as $\nu(\nu+1) \to 0$, reflecting the vanishing of $H_r$ at the boundaries of the physical parameter range. For the zonal case $m = 0$, an alternative definition using the ratio $E_\phi/H_\theta$ proves more useful:
\begin{equation}
Z_{\text{TE}}^{(m=0)} = \frac{E_\phi}{H_\theta} = i\omega\mu r \cdot \frac{j_\nu(kr)}{[rj_\nu(kr)]'/r}.
\label{eq:Z_TE_zonal}
\end{equation}

For TM modes, the corresponding impedance is
\begin{equation}
Z_{\text{TM}} = \frac{E_r}{H_\theta} = {-\frac{\nu(\nu+1)}{i\omega\varepsilon m r \sin\theta}},
\label{eq:Z_TM}
\end{equation}
with the zonal alternative
\begin{equation}
Z_{\text{TM}}^{(m=0)} = \frac{E_\theta}{H_\phi} = \frac{[rj_\nu(kr)]'/r}{i\omega\varepsilon j_\nu(kr)}.
\label{eq:Z_TM_zonal}
\end{equation}

A fundamental duality connects the TE and TM impedances. From equations~\eqref{eq:Z_TE} and~\eqref{eq:Z_TM}, one verifies directly that
\begin{equation}
Z_{\text{TE}} \cdot Z_{\text{TM}} = -\eta^2,
\label{eq:impedance_duality}
\end{equation}
where $\eta = \sqrt{\mu/\varepsilon}$ is the intrinsic impedance of the medium. This relation, familiar from waveguide theory~\cite{Collin1991, Pozar2012}, expresses the electromagnetic duality between TE and TM polarizations and provides a useful check on the field expressions.

\section{Boundary Modifications and Self-Adjoint Extensions}
\label{boundary}
\subsection{The Role of Domain Geometry}
The spectral structure established depends critically on the assumption that the angular domain encompasses the full sphere, with regularity demanded at both poles. Physical modifications to the cavity geometry---the insertion of conducting cones or wedges---alter the domain on which the angular equation is posed and thereby change the spectral properties of the electromagnetic modes. The mathematical framework for understanding these changes is the theory of self-adjoint extensions, which describes how boundary conditions at singular points of a differential operator determine its spectral characteristics.

The angular Laplacian appearing in the associated Legendre equation possesses regular singular points at $\theta = 0$ and $\theta = \pi$. On the full interval $(0, \pi)$, demanding regularity at both endpoints constitutes a choice of self-adjoint extension that yields the discrete spectrum familiar from standard treatments. Alternative choices, corresponding to different physical boundary conditions, select different self-adjoint extensions with qualitatively different spectra. This perspective, developed rigorously in the monographs of Reed and Simon~\cite{ReedSimon1975} and Zettl~\cite{Zettl2005}, reveals that spectral discreteness is not an intrinsic property of the differential equation but rather a consequence of the boundary conditions imposed. Figure~\ref{fig:geometries} illustrates the various cavity geometries considered in this section and their effects on the spectral parameters.

\begin{figure}[htbp]
\centering

\begin{tikzpicture}[
  scale=0.6,
  sphere/.style={fill=gray!35, draw=black, thick},
  cut/.style={draw=black, very thick},
  cone/.style={fill=black, opacity=0.25},
  hatched/.style={pattern=north east lines, pattern color=gray!60},
  axis/.style={thick, ->},
  eq/.style={gray, dashed},
]

\begin{scope}[shift={(0,0)}]
  \path[sphere] (0,0) circle (2);
  \draw[eq] (-2,0) -- (2,0);

  \draw[axis] (0,-2.7) -- (0,2.9) node[above] {\large $z$};

  \fill (0,2) circle (3pt) node[above right] {N};
  \fill (0,-2) circle (3pt) node[below right] {S};

  \draw[axis] (0,0) -- (2.5,0) node[right] {$\rho$};

  \node[below] at (0,-3.3) {\large\textbf{(a) Full sphere}};
  \node[below] at (0,-4.0) {\large $\nu, m \in \mathbb{Z}$};
\end{scope}

\begin{scope}[shift={(7,0)}]
  \draw[thick, dashed, gray] (0,0) circle (2);

  \fill[gray!35] (0,0) -- (2,0) arc (0:270:2) -- cycle;
  \draw[thick] (2,0) arc (0:270:2);

  \draw[cut] (0,0) -- (2,0);
  \draw[cut] (0,0) -- (0,-2);

  \draw[thick, <->] (1.0,0) arc (0:270:1.0);
  \node at (-0.7,0.7) {\Large $\Phi$};

  \fill[hatched] (0,0) -- (2,0) arc (0:-90:2) -- cycle;
  \node[gray] at (1.0,-0.7) {\scriptsize removed};

  \fill (0,0) circle (2pt);
  \draw[axis] (0,0) -- (2.7,0) node[right] {$x$};
  \draw[axis] (0,0) -- (0,2.7) node[above] {$y$};
  \node at (0.25,0.25) {$\odot z$};

  \node[below] at (0,-3.3) {\large\textbf{(b) Azimuthal wedge}};
  \node[below] at (0,-4.0) {\large $m = n\pi/\Phi$};
\end{scope}

\end{tikzpicture}

\vspace{0.8cm}

\begin{tikzpicture}[
  scale=0.6,
  sphere/.style={fill=gray!35, draw=black, thick},
  cone/.style={fill=black, opacity=0.25},
  axis/.style={thick, ->},
  eq/.style={gray, dashed},
]

\begin{scope}[shift={(0,0)}]
  \path[sphere] (0,0) circle (2);

  \begin{scope}
    \clip (0,0) circle (2);
    \fill[cone] (0,0) -- (-0.9,3) -- (0.9,3) -- cycle;
  \end{scope}

  \draw[thick] (0,0) -- (-0.6,2);
  \draw[thick] (0,0) -- (0.6,2);
  \draw[thick] (-0.6,2) arc (108:72:2);

  \draw[thick, <->] (0,1.2) arc (90:72:1.2);
  \node at (0.45,1.4) {\large $\theta_c$};

  \draw[eq] (-2,0) -- (2,0);
  \draw[axis] (0,-2.7) -- (0,2.9) node[above] {\large $z$};

  \fill (0,-2) circle (3pt) node[below right] {S};

  \node[below] at (0,-3.3) {\large\textbf{(c) Single cone}};
  \node[below] at (0,-4.0) {\large $\nu = \nu(\theta_c)$};
\end{scope}

\begin{scope}[shift={(7,0)}]
  \path[sphere] (0,0) circle (2);

  \begin{scope}
    \clip (0,0) circle (2);
    \fill[cone] (0,0) -- (-0.7,3) -- (0.7,3) -- cycle;
    \fill[cone] (0,0) -- (-0.7,-3) -- (0.7,-3) -- cycle;
  \end{scope}

  \draw[thick] (0,0) -- (-0.47,2);
  \draw[thick] (0,0) -- (0.47,2);
  \draw[thick] (-0.47,2) arc (103:77:2);

  \draw[thick] (0,0) -- (-0.47,-2);
  \draw[thick] (0,0) -- (0.47,-2);
  \draw[thick] (-0.47,-2) arc (257:283:2);

  \draw[eq] (-2,0) -- (2,0);
  \draw[axis] (0,-2.7) -- (0,2.9) node[above] {\large $z$};

  \node[below] at (0,-3.3) {\large\textbf{(d) Biconical}};
  \node[below] at (0,-4.0) {\large (double cone)};
\end{scope}

\end{tikzpicture}

\vspace{0.8cm}

\begin{tikzpicture}[
  scale=0.6,
  axis/.style={thick, ->},
  eq/.style={gray, dashed},
  cut/.style={draw=black, very thick},
  cone/.style={fill=black, opacity=0.25},
  hatched/.style={pattern=north east lines, pattern color=gray!60},
]

\begin{scope}[shift={(0,0)}]
  \fill[gray!35] (0,2) arc(90:270:2) -- cycle;
  \draw[thick] (0,2) arc(90:270:2);

  \fill[gray!15] (0,-2) -- (0,2) -- (2,0) -- cycle;
  \draw[cut] (0,-2) -- (0,2);

  \begin{scope}
    \clip (0,2) arc(90:270:2) -- cycle;
    \fill[cone] (0,0) -- (-0.6,3) -- (0,3) -- cycle;
  \end{scope}
  \draw[thick] (0,0) -- (-0.4,2);

  \fill[hatched] (0,2) arc(90:0:2) -- (0,0) -- cycle;

  \draw[eq] (-2,0) -- (0,0);

  \draw[axis] (0,-2.7) -- (0,2.9) node[above] {\large $z$};
  \draw[axis] (0,0) -- (2.5,0) node[right] {$x$};

  \node[below] at (0,-3.3) {\large\textbf{(e) Cone + wedge}};
  \node[below] at (0,-4.0) {\large $\nu(\theta_c),\; m = n\pi/\Phi$};
\end{scope}

\begin{scope}[shift={(7,0)}]
  \fill[gray!35] (0,2) arc(90:270:2) -- cycle;
  \draw[thick] (0,2) arc(90:270:2);

  \fill[gray!15] (0,-2) -- (0,2) -- (2,0) -- cycle;
  \draw[cut] (0,-2) -- (0,2);

  \begin{scope}
    \clip (0,2) arc(90:270:2) -- cycle;
    \fill[cone] (0,0) -- (-0.5,3) -- (0,3) -- cycle;
    \fill[cone] (0,0) -- (-0.5,-3) -- (0,-3) -- cycle;
  \end{scope}
  \draw[thick] (0,0) -- (-0.35,2);
  \draw[thick] (0,0) -- (-0.35,-2);

  \fill[hatched] (0,2) arc(90:0:2) -- (0,0) -- cycle;

  \draw[eq] (-2,0) -- (0,0);

  \draw[axis] (0,-2.7) -- (0,2.9) node[above] {\large $z$};
  \draw[axis] (0,0) -- (2.5,0) node[right] {$x$};

  \node[below] at (0,-3.3) {\large\textbf{(f) Bicone + wedge}};
  \node[below] at (0,-4.0) {\large Full $(\nu, m)$ access};
\end{scope}

\end{tikzpicture}

\caption{Spherical cavity geometries and their spectral characteristics, shown as cross-sections.
(a) Full sphere (meridional view): both $\nu$ and $m$ restricted to integers.
(b) Azimuthal wedge (equatorial view from above): removing the angular region $\phi > \Phi$ permits non-integer $m = n\pi/\Phi$; hatched region shows removed portion.
(c) Single cone at half-angle $\theta_c$: removes the polar cap and allows continuous $\nu(\theta_c) > 0$.
(d) Double cone (biconical): symmetric cones at both poles.
(e) Cone with wedge: independent control of $\nu$ via cone angle and $m$ via wedge angle; hatched shows removed azimuthal section.
(f) Biconical with wedge: full two-parameter access to continuous $(\nu, m)$.}
\label{fig:geometries}
\end{figure}
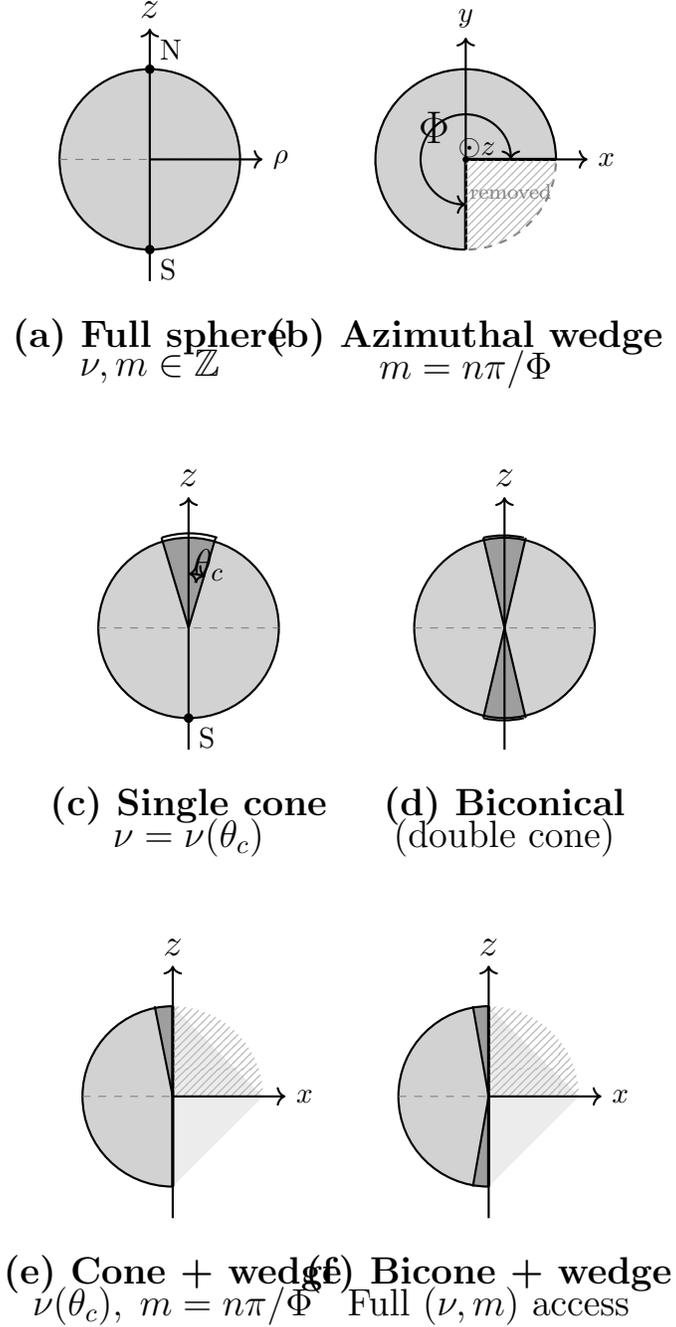

\subsection{Conical Boundaries}

Consider a perfectly conducting cone of half-angle $\theta_c$, coaxial with the polar axis, inserted into the spherical cavity. The cone removes from the domain all points with $\theta > \theta_c$, replacing the regularity condition at $\theta = \pi$ with a boundary condition at $\theta = \theta_c$. The electromagnetic boundary conditions require vanishing of the tangential electric field at the conical surface.

For TM modes, the condition $E_r(\theta_c) = 0$ requires the angular function to vanish at the cone:
\begin{equation}
\Theta_\nu^m(\theta_c) = 0.
\label{eq:cone_BC_TM}
\end{equation}
This equation, together with the regularity requirement at $\theta = 0$, determines the admissible values of $\nu$ as a function of the cone angle $\theta_c$ and the azimuthal index $m$. For the zonal case $m = 0$, equation~\eqref{eq:cone_BC_TM} becomes $P_\nu(\cos\theta_c) = 0$, the zeros of the Legendre function. As $\theta_c$ varies continuously from $0$ to $\pi$, the roots $\nu$ vary continuously, converting the isolated integer spectrum of the full sphere into a continuous family parameterized by the cone geometry.

For TE modes, the condition $E_\phi(\theta_c) = 0$ requires
\begin{equation}
\left.\frac{d\Theta_\nu^m}{d\theta}\right|_{\theta = \theta_c} = 0.
\label{eq:cone_BC_TE}
\end{equation}
In the zonal case, this condition becomes $P'_{\nu}(\cos\theta_c)=0$, corresponding to extrema of the Legendre function. For the cap-removal geometry studied here, decreasing $\theta_c$ toward zero yields progressively smaller values of $\nu$, approaching the lower end of the continuous branch $0<\nu<1$.

\subsection{Azimuthal Wedges}

A conducting wedge spanning the angular range $0 < \phi < \Phi$, where $\Phi < 2\pi$, restricts the azimuthal domain and modifies the quantization of the index $m$. The boundary conditions $E_\text{tan} = 0$ at $\phi = 0$ and $\phi = \Phi$ select standing-wave solutions in azimuth rather than the traveling-wave form $e^{im\phi}$. The admissible values of the azimuthal index become
\begin{equation}
m = \frac{n\pi}{\Phi}, \quad n = 1, 2, 3, \ldots
\label{eq:wedge_m}
\end{equation}
Since the angular equation~\eqref{eq:angular_ode} depends on $m$ only through $m^2$, solutions for $\pm m$ yield identical angular eigenvalues. We adopt the convention $m > 0$ without loss of generality; the physical fields are real combinations of $e^{\pm im\phi}$, realized in the wedge geometry as standing waves $\sin(m\phi)$ or $\cos(m\phi)$ depending on the boundary conditions at the wedge faces.

For wedge angles that are not rational multiples of $\pi$, the values of $m$ are non-integer. On the sectoral line $\nu = m$, such non-integer values correspond to points on the continuous dispersion curve between the standard integer resonances. The wedge geometry thus provides physical access to the continuous sectoral branch, with the wedge angle serving as a tuning parameter that selects specific points along the dispersion curve.

The combination of a conical boundary at some $\theta_c < \pi$ with an azimuthal wedge of opening $\Phi < 2\pi$ provides access to a two-parameter family of modes labeled by continuous values of both $\nu$ and $m$. This configuration, while geometrically complex, represents the most general physical realization of the continuous spectral structure underlying the spherical cavity modes.

\section{Energy Considerations and Power Flow}

\label{energy}
The physical admissibility of an electromagnetic mode requires that the total field energy within the cavity remain finite. This condition places constraints on the behavior of the fields near singular points of the geometry. Beyond mere admissibility, the spatial distribution of energy and the structure of power flow provide physical insight into the nature of the modes and their connection to practical applications. The time-averaged electromagnetic energy stored in the cavity is
\begin{equation}
U = \frac{1}{4} \int_V \left( \epsilon_0 |\mathrm{E}|^2 + \mu_0 |\mathrm{H}|^2 \right) dV,
\label{eq:total_energy}
\end{equation}
where the factor of $1/4$ (rather than $1/2$) accounts for time-averaging of sinusoidal fields. In spherical coordinates, the volume element is $dV = r^2 \sin\theta \, dr \, d\theta \, d\phi$, and the integral factorizes when the fields admit separated solutions. For a mode with radial dependence $j_\nu(kr)$, angular dependence $P_\nu^m(\cos\theta)$, and azimuthal dependence $e^{im\phi}$, the energy integral takes the form
\begin{equation}
U = U_0 \cdot I_r(\nu) \cdot I_\theta(\nu, m) \cdot I_\phi(m),
\label{eq:energy_factored}
\end{equation}
where $U_0$ contains dimensional factors and $I_r$, $I_\theta$, $I_\phi$ are dimensionless integrals over the respective coordinates.

The relevant integrability criterion is finite electromagnetic energy, not merely square integrability of the angular eigenfunction under the Sturm--Liouville weight $\sin\theta$. Since the energy density involves products of field components, each containing the angular function and potentially factors of $1/\sin\theta$ arising from the spherical coordinate representation, the physical constraint is more restrictive than abstract $L^{2}$ membership.

\subsection{Radial Energy Convergence}
\label{subsec:radial_energy}

The electromagnetic energy stored in a cavity mode must be finite for the mode to be physically realizable. This requirement constrains the admissible values of the angular eigenvalue $\nu$ through the behavior of the fields near the origin.

From the field expressions derived in Section~\ref{EM}, the radial electric field component for TM modes is
\begin{equation}
E_r = \frac{\nu(\nu+1)}{r} A j_\nu(kr) \Theta_\nu^m(\theta) e^{im\phi}.
\label{eq:Er_near_origin}
\end{equation}
Near the origin, the spherical Bessel function admits the asymptotic expansion $j_\nu(kr) \sim (kr)^\nu/(2\nu+1)!!$ for $\nu > -1/2$, yielding
\begin{equation}
E_r \sim \nu(\nu+1) \cdot r^{\nu-1} \quad \text{as } r \to 0.
\label{eq:Er_scaling}
\end{equation}
The tangential components $E_\theta$ and $E_\phi$ involve the combination $[rj_\nu(kr)]'/r$, which behaves as $r^{\nu-1}$ near the origin---the same scaling as the radial component. For TE modes, the magnetic field components exhibit identical radial behavior by electromagnetic duality.

The energy density scales as $|E|^2 \sim r^{2\nu-2}$. Including the spherical volume element $r^2\,dr\,d\Omega$, the radial contribution to the energy integral near the origin becomes
\begin{equation}
\int_0^\epsilon |E|^2 r^2 \, dr \sim \int_0^\epsilon r^{2\nu} \, dr 
= \frac{\epsilon^{2\nu+1}}{2\nu+1}.
\label{eq:radial_energy_integral}
\end{equation}
This integral converges if and only if $2\nu + 1 > 0$, yielding the energy integrability condition
\begin{equation}
\boxed{\nu > -\frac{1}{2}}
\label{eq:nu_bound}
\end{equation}
for both TE and TM modes.

This condition is remarkably permissive: it admits all positive values of $\nu$,
including the fractional values that arise from wedge and cone geometries. For example,
the PEC/PMC wedge with $270^\circ$ opening supports modes with $\nu = m = 1/3$, which
satisfies~\eqref{eq:nu_bound} comfortably. Similarly, shallow cone geometries can produce
eigenvalues approaching zero from above without violating energy integrability.

\begin{remark}[Consistency with numerical simulations]
The bound $\nu > -1/2$ is consistent with all modes observed in HFSS eigenmode simulations of wedge and cone geometries. In particular, the $\nu = m = 1/3$ mode of the PEC/PMC wedge (observed at 6.27~GHz for a 15~mm radius cavity) has finite, well-defined energy despite the field singularity at the wedge edge. The variational formulation employed by finite-element solvers naturally handles such edge singularities through the weak form of Maxwell's equations, which involves field curls rather than the fields themselves.
\end{remark}

\begin{remark}[Spherical Neumann function exclusion]
The radial equation admits two linearly independent solutions: the spherical Bessel function $j_\nu(kr) \sim r^\nu$ and the spherical Neumann function
$y_\nu(kr) \sim r^{-\nu-1}$ as $r \to 0$. The electromagnetic energy near the origin converges only if the radial field singularity is integrable.
For $j_\nu$, this requires $\nu > -\tfrac{1}{2}$; for $y_\nu$, this requires $\nu < -\tfrac{1}{2}$.

Since the eigenvalue $\lambda = \nu(\nu+1)$ is invariant under $\nu \to -1-\nu$, each physical eigenvalue admits two parameter values related by this symmetry. By choosing the representative with $\nu > -\tfrac{1}{2}$, the Bessel function $j_\nu$ is always the finite-energy solution. This justifies both the restriction $\nu > -\tfrac{1}{2}$ and the choice
$R(r) = j_\nu(kr)$ in Section~\ref{math}.

\end{remark}

\subsection{Angular Energy Convergence}
The angular energy integral requires examination near both poles, where the coordinate system becomes singular. We treat the sectoral, tesseral, and zonal cases separately.
\subsubsection{Sectoral Modes ($\nu = m$)}
For sectoral modes, the angular function is $\Theta(\theta) = \sin^m\theta$, which vanishes at both poles as $\theta^m$ (near $\theta = 0$) and $(\pi - \theta)^m$ (near $\theta = \pi$). The angular energy integral becomes
\begin{equation}
I_\theta = \int_0^\pi |\Theta(\theta)|^2 \sin\theta \, d\theta = \int_0^\pi \sin^{2m+1}\theta \, d\theta.
\label{eq:angular_integral_sectoral}
\end{equation}
Near $\theta = 0$, the integrand behaves as $\theta^{2m+1}$, yielding
\begin{equation}
\int_0^\epsilon \theta^{2m+1} \, d\theta = \frac{\epsilon^{2m+2}}{2m+2},
\label{eq:angular_near_pole}
\end{equation}
which converges for all $m > -1$. The analogous integral near $\theta = \pi$ converges under the same condition. For all $m > 0$, the sectoral angular energy integral is therefore manifestly finite. The complete integral can be evaluated in closed form using the beta function:
\begin{equation}
\int_0^\pi \sin^{2m+1}\theta \, d\theta = B\left(m+1, \frac{1}{2}\right) = \frac{\sqrt{\pi} \, \Gamma(m+1)}{\Gamma(m + 3/2)} = \frac{2^{2m+1} (m!)^2}{(2m+1)!},
\label{eq:angular_integral_closed}
\end{equation}
valid for $m > -1$. The final equality holds only for integer $m$; for non-integer $m$, the gamma function expression remains valid.

While the angular function $\sin^m\theta$ itself is well-behaved, certain field components contain factors of $1/\sin\theta$ from the spherical coordinate representation. For TM modes, the azimuthal magnetic field component includes
\begin{equation}
H_\phi \propto \frac{m}{\sin\theta} P_\nu^m(\cos\theta).
\label{eq:H_phi_factor}
\end{equation}
For sectoral modes with $P_m^m(\cos\theta) \propto \sin^m\theta$, this becomes
\begin{equation}
H_\phi \propto \frac{m \sin^m\theta}{\sin\theta} = m \sin^{m-1}\theta.
\label{eq:H_phi_sectoral}
\end{equation}
Near $\theta = 0$, the contribution to the energy density scales as $|H_\phi|^2 \sim m^2 \theta^{2m-2}$. Including the volume element:
\begin{equation}
\int_0^\epsilon m^2 \theta^{2m-2} \cdot \theta \, d\theta = m^2 \int_0^\epsilon \theta^{2m-1} \, d\theta = \frac{m^2 \epsilon^{2m}}{2m} = \frac{m \epsilon^{2m}}{2}.
\label{eq:H_phi_energy}
\end{equation}
This converges for all $m > 0$ and vanishes as $m \to 0^+$, so no additional constraint arises from the $1/\sin\theta$ factors in the field expressions.

\subsubsection{Tesseral Modes ($\nu = m + k$, $k \geq 1$)}
For tesseral modes with $\nu = m + k$ where $k$ is a positive integer, the angular function takes the form
\begin{equation}
\Theta(\theta) = \sin^m\theta \cdot Q_k(\cos\theta),
\label{eq:angular_tesseral}
\end{equation}
where $Q_k$ is a polynomial of degree $k$ arising from the terminating hypergeometric series (see Eq.~\eqref{eq:theta_hypergeometric}). Since $Q_k(\cos\theta)$ is bounded on $[0, \pi]$, the behavior near the poles is dominated by the $\sin^m\theta$ factor, exactly as in the sectoral case. The angular energy integral therefore converges under the same condition $m > 0$.

The field components for tesseral modes contain the same $1/\sin\theta$ factors as sectoral modes. Near the poles:
\begin{equation}
H_\phi \propto \frac{m}{\sin\theta} P_\nu^m(\cos\theta) \sim m \sin^{m-1}\theta \cdot Q_k(\cos\theta),
\end{equation}
which remains integrable for $m > 0$ since $Q_k$ is bounded. Thus tesseral modes satisfy the same energy integrability condition as sectoral modes: $m > 0$.

\subsubsection{Zonal Modes ($m = 0$)}
For zonal modes, the angular function is the Legendre function $P_\nu(\cos\theta)$, which is bounded on $[0, \pi]$ for all $\nu \geq 0$. The field components containing $m/\sin\theta$ factors vanish identically when $m = 0$, eliminating potential singularities For integer $\nu$ on the full sphere, the angular energy integral
\begin{equation}
I_\theta = \int_0^\pi |P_\nu(\cos\theta)|^2 \sin\theta \, d\theta = \frac{2}{2\nu + 1}
\label{eq:angular_integral_zonal}
\end{equation}
follows from the orthogonality of Legendre polynomials and converges for all $\nu > -1/2$, a condition satisfied by all physical zonal modes.

\begin{remark}[Zonal modes with non-integer $\nu$]
For cone geometries where the polar domain is truncated to $\theta \in [\theta_c, \pi]$, zonal modes can have non-integer $\nu$ determined by the boundary condition $P_\nu(\cos\theta_c) = 0$. In this case, Eq.~\eqref{eq:angular_integral_zonal} does not apply directly, but the Legendre function $P_\nu(\cos\theta)$ remains bounded on the truncated domain, ensuring finite energy. The normalization integral must be evaluated numerically for non-integer $\nu$.
\end{remark}

\subsubsection{TM Modes}

For TM modes, $H_r = 0$ by definition, so the Poynting vector simplifies to
\begin{align}
S_r^{\mathrm{TM}} &= \frac{1}{2} \mathrm{Re}\left( E_\theta H_\phi^* - E_\phi H_\theta^* \right), \\
S_\theta^{\mathrm{TM}} &= -\frac{1}{2} \mathrm{Re}\left( E_r H_\phi^* \right), \\
S_\phi^{\mathrm{TM}} &= \frac{1}{2} \mathrm{Re}\left( E_r H_\theta^* \right).
\label{eq:poynting_TM}
\end{align}
Substituting the field expressions from Section~\ref{EM}, the azimuthal power flow for TM modes becomes
\begin{equation}
S_\phi^{\mathrm{TM}} = \frac{1}{2\eta} |E_r|^2 \cdot \frac{1}{Z_{\mathrm{TM}}},
\label{eq:S_phi_TM_explicit}
\end{equation}
where $Z_{\mathrm{TM}}$ is the TM wave impedance derived in Section~\ref{EM}. The azimuthal component $S_\phi$ represents energy circulation around the cavity axis and is nonzero for all modes with $m \neq 0$.

\subsubsection{TE Modes}
For TE modes, $E_r = 0$, yielding
\begin{align}
S_r^{\mathrm{TE}} &= \frac{1}{2} \mathrm{Re}\left( E_\theta H_\phi^* - E_\phi H_\theta^* \right), \\
S_\theta^{\mathrm{TE}} &= \frac{1}{2} \mathrm{Re}\left( E_\phi H_r^* \right), \\
S_\phi^{\mathrm{TE}} &= -\frac{1}{2} \mathrm{Re}\left( E_\theta H_r^* \right).
\label{eq:poynting_TE}
\end{align}
The corresponding expression for azimuthal power flow is
\begin{equation}
S_\phi^{\mathrm{TE}} = \frac{\eta}{2} |H_r|^2 \cdot Z_{\mathrm{TE}},
\label{eq:S_phi_TE_explicit}
\end{equation}
involving the TE wave impedance.

\subsubsection{Azimuthal Power Flow and Wave Impedance}
The total power flowing in the azimuthal direction through a surface of constant $\phi$ is obtained by integrating the Poynting vector:
\begin{equation}
P_\phi = \int_0^a \int_0^\pi S_\phi \, r \, dr \, r \sin\theta \, d\theta = \int_0^a \int_0^\pi S_\phi \, r^2 \sin\theta \, dr \, d\theta.
\label{eq:P_phi_integral}
\end{equation}
For traveling-wave modes (real $m$ permitted by non-axisymmetric geometries), this integral yields a finite, nonzero result proportional to $m$.
The wave impedances $Z_{\mathrm{TM}}$ and $Z_{\mathrm{TE}}$ characterize the ratio of transverse field components and satisfy the duality relation
\begin{equation}
Z_{\mathrm{TE}} \cdot Z_{\mathrm{TM}} = -\eta^2,
\label{eq:impedance_duality}
\end{equation}
as derived in Section~\ref{EM}. This relation ensures that TE and TM modes with the same angular indices carry equal power for equal field amplitudes.

\subsubsection{Standing vs. Traveling Waves}
For integer $m$, the azimuthal dependence $e^{im\phi}$ can be decomposed into standing waves $\cos(m\phi)$ and $\sin(m\phi)$. Standing-wave modes have $\langle S_\phi \rangle = 0$ when averaged over azimuth, representing purely oscillatory energy exchange rather than net circulation.

For non-integer $m$ realized through wedge geometries, the boundary conditions at the wedge faces impose a specific superposition of $e^{+im\phi}$ and $e^{-im\phi}$ components. The resulting standing-wave pattern is confined to the angular sector $0 < \phi < \Delta\phi$, with the Poynting vector exhibiting nodes at the conducting boundaries.

\subsection{Energy Distribution and Mode Localization}
The spatial distribution of electromagnetic energy provides insight into mode character. For sectoral modes with $\nu = m$, the angular factor $\sin^{2m}\theta$ concentrates the energy toward the equatorial plane as $m$ increases. The energy density at the equator ($\theta = \pi/2$) relative to that at latitude $\theta$ scales as
\begin{equation}
\frac{u(\theta)}{u(\pi/2)} = \sin^{2m}\theta,
\label{eq:energy_latitude}
\end{equation}
which for large $m$ becomes sharply peaked near $\theta = \pi/2$. This equatorial concentration is characteristic of whispering-gallery-type behavior, though in the azimuthal rather than polar direction. Conversely, for small $m \to 0^+$, the angular distribution becomes nearly 
uniform. However, this limit must be interpreted carefully: while the 
normalized angular profile $\sin^{2m}\theta \to 1$ becomes isotropic, the 
overall field amplitude vanishes as $m \to 0^+$, and the electromagnetic 
field is identically zero at $(\nu, m) = (0,0)$ (see Section~\ref{sec:null}). The transition from equatorially concentrated ($m \gg 1$) to uniform ($m \to 0^+$) angular distributions thus occurs along a family of modes whose energy tends to zero in the limit.

\subsection{Connection to Feed Structures}
This distinction is relevant for antenna feed design~\cite{Balanis2016, Chu1948}. Conical antenna structures fed at their apex excite modes with small but nonzero effective $m$, determined by the feed geometry. The near-field pattern of such structures is well described by the sectoral solutions with continuous $m$, providing a theoretical foundation for the design of broadband feeds and transitions~\cite{Schelkunoff1943, Papas1950}.

The limiting case $m \to 0^+$ along the sectoral branch requires careful 
interpretation. While the angular function $\sin^m\theta \to 1$ as $m \to 0^+$, and the Debye potential $\Pi = j_0(kr)$ remains well-defined at $(0,0)$, the electromagnetic field extracted from this potential vanishes identically (see Section~\ref{sec:null}). The approach $m \to 0^+$ thus represents a sequence of propagating modes whose field amplitude tends to zero, rather than a smooth connection to a nontrivial zonal mode. Physical antenna structures necessarily operate at small but finite $m > 0$.

The fundamental limit on antenna size derived by Chu~\cite{Chu1948} and elaborated by Harrington~\cite{Harrington1960} can be understood in terms of the minimum $m$ (and hence minimum $\nu$) supportable by a spherical volume of given radius. Smaller antennas require modes with smaller angular indices, which in turn have lower $Q$ factors and broader bandwidths but also lower radiation efficiency.

\section{Numerical Validation}
\label{validation}
The preceding analysis predicts that electromagnetic resonances in modified spherical cavities are governed by continuous angular indices $(\nu, m)$ rather than the integer quantum numbers $(\ell, m)$ of the full sphere. This section validates this framework through finite-element eigenmode simulations using Ansys HFSS, demonstrating quantitative agreement between theory and simulation across wedge-only, cone-only, and combined configurations.

All simulations use a spherical vacuum cavity of radius $a = 15~\mathrm{mm}$ with perfect electric conductor (PEC) boundary conditions. Azimuthal wedges are implemented by inserting PEC sheets along radial half-planes intersecting the polar axis, restricting the domain to $\phi \in [0, \Delta\phi]$. Polar cones are created by Boolean subtraction: a cone of height $h = a$ and outer radius $r_{\mathrm{out}}$ is removed from the north pole, replacing the polar cap with a conical PEC boundary at half-angle $\theta_c = \arctan(r_{\mathrm{out}}/a)$. No metallic structures remain inside the cavity; all boundaries are part of the PEC cavity wall. Mesh refinement ensures eigenfrequency convergence to better than $0.1\%$.

The validation proceeds by comparing HFSS eigenfrequencies against theoretical predictions computed entirely from the boundary-value problem, with no fitting parameters. For each geometry, we solve the appropriate angular eigenvalue problem to determine $(\nu, m)$, then apply radial quantization to obtain the resonant frequencies.

\subsection{Wedge Geometries: Non-Integer Azimuthal Index}
\label{subsec:wedge_validation}

In this section, we define the \emph{wedge angle} $\alpha$ as the angular extent of the PEC sheet(s) inserted into the cavity, and the \emph{cavity opening} $\Phi = 2\pi - \alpha$ as the remaining azimuthal domain. Thus a ``$90^\circ$ wedge'' refers to a geometry with $\alpha = 90^\circ$ of PEC material removed, leaving a cavity opening of $\Phi = 270^\circ$.

For a cavity opening $\Phi$ with PEC boundaries on both faces, the azimuthal eigenfunctions take the form $\Phi(\phi) = \sin(m\phi)$, where the boundary condition $\Phi(\Phi) = 0$ quantizes the azimuthal index as
\begin{equation}
\label{eq:m_quantization}
m = \frac{n\pi}{\Phi}, \qquad n = 1, 2, 3, \ldots
\end{equation}
This quantization is exact and follows directly from the separable structure of Maxwell's equations in spherical coordinates. For cavity openings that are not submultiples of $\pi$, the values of $m$ are non-integer.

Since the polar domain remains $\theta \in (0, \pi)$, regularity at both poles is still required. For a given $m > 0$ (integer or not), the solution regular at the north pole takes the form
\begin{equation}
\Theta(\theta) = \sin^m\theta \cdot {}_2F_1\left(m-\nu, m+\nu+1; m+1; \sin^2(\theta/2)\right).
\label{eq:theta_hypergeometric}
\end{equation}
As $\theta \to \pi$, the argument $\sin^2(\theta/2) \to 1$, where the hypergeometric function diverges unless the series terminates. Termination occurs when $m - \nu = -k$ for non-negative integer $k$, yielding the eigenvalue condition
\begin{equation}
\nu = m + k, \qquad k = 0, 1, 2, \ldots
\label{eq:nu_spectrum_wedge}
\end{equation}
The case $k = 0$ gives the sectoral mode $\nu = m$; the cases $k \geq 1$ give tesseral modes with non-integer $\nu$ whenever $m$ is non-integer.

\begin{figure}[t]
\centering
\includegraphics[width=\textwidth]{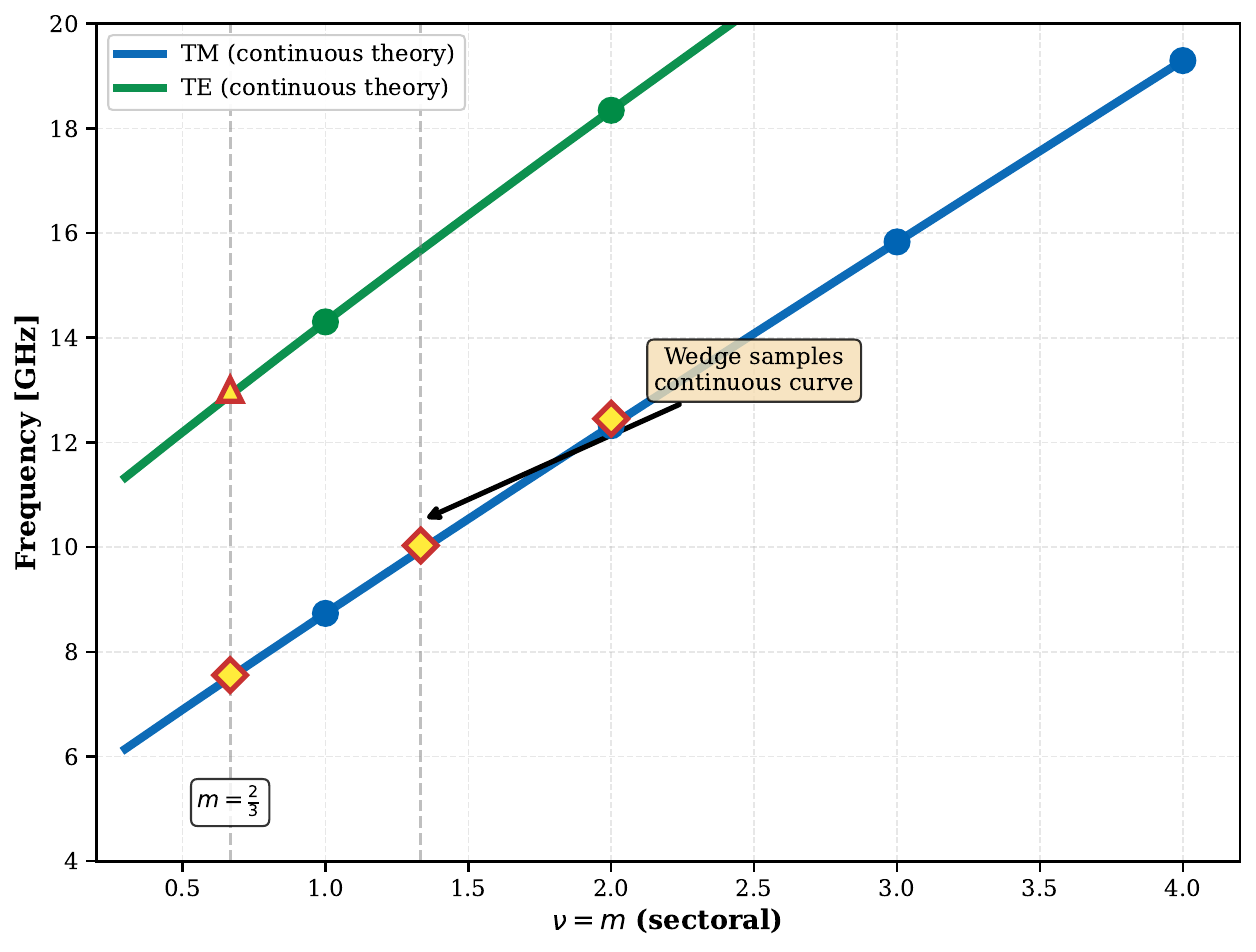}
\caption{Fundamental TM frequency versus blue cavity opening angle $\Phi$.
The minimum near $\Phi = 270^\circ$ (corresponding to a $90^\circ$ wedge) arises because this geometry permits $m = 2/3 < 1$, sampling the continuous dispersion curve below the first integer mode.
The $\Phi = 180^\circ$ geometry (hemisphere) enforces $m \geq 1$, yielding higher frequencies despite smaller volume.
Theory (solid line with circles) computed from Eq.~\eqref{eq:m_quantization} and radial TM quantization; HFSS data (squares) agrees within $1\%$.}
\label{fig:wedge_angle}
\end{figure}

\begin{table}[htbp]
\centering
\caption{Eigenmode spectrum for a spherical cavity with a $90^\circ$ wedge
($\Phi = 270^\circ$ cavity opening, $a = 15~\mathrm{mm}$).
The azimuthal index $m = 2n/3$ follows from Eq.~\eqref{eq:m_quantization};
the angular eigenvalue $\nu = m + k$ for non-negative integer $k$,
per Eq.~\eqref{eq:nu_spectrum_wedge}.
Modes with $k = 0$ are sectoral, while modes with $k \geq 1$ are tesseral.
All modes agree within $1.4\%$.}
\label{tab:wedge_spectrum}

\begin{tabular}{cccccccc}
\toprule
Mode & Type & $\nu$ & $m$ & $k$ &
$f_{\mathrm{theory}}$ [GHz] &
$f_{\mathrm{HFSS}}$ [GHz] &
Error [\%] \\
\midrule
1 & TM & $2/3$ & $2/3$ & 0 & 7.51 & 7.57 & $-0.8$ \\
2 & TM & $4/3$ & $4/3$ & 0 & 9.93 & 10.03 & $-1.0$ \\
3 & TM & $5/3$ & $2/3$ & 1 & 11.14 & 11.25 & $-1.0$ \\
4 & TM & $2$   & $2$   & 0 & 12.32 & 12.44 & $-1.0$ \\
5 & TE & $2/3$ & $2/3$ & 0 & 12.89 & 13.02 & $-1.0$ \\
6 & TM & $7/3$ & $4/3$ & 1 & 13.59 & 13.62 & $-0.2$ \\
\bottomrule
\end{tabular}
\end{table}

The key test case is the $90^\circ$ wedge (cavity opening $\Phi = 3\pi/2 = 270^\circ$), which yields $m = 2n/3$ for $n = 1, 2, 3, \ldots$ The fundamental mode has $m = 2/3$---a value strictly forbidden on the full sphere, where regularity of $e^{im\phi}$ under $\phi \to \phi + 2\pi$ requires integer $m$. For sectoral TM modes with $\nu = m$, the resonant frequency is determined by the first root of $[x j_\nu(x)]' = 0$, giving $x'_{2/3,1} \approx 2.36$ and
\begin{equation}
f_1 = \frac{c \, x'_{2/3,1}}{2\pi a} \approx 7.51~\mathrm{GHz}.
\end{equation}
HFSS returns $f_1 = 7.57~\mathrm{GHz}$, an agreement within $0.8\%$.

This result has notable implications. The $90^\circ$ wedge ($\Phi = 270^\circ$) possesses a \emph{larger} cavity volume than the hemisphere ($\Phi = 180^\circ$), yet exhibits a \emph{lower} fundamental frequency. The hemisphere enforces integer $m = 1, 2, 3, \ldots$, yielding a fundamental TM frequency of $8.73~\mathrm{GHz}$ (theory) versus $8.81~\mathrm{GHz}$ (HFSS)---some 17\% higher than the $90^\circ$ wedge case despite the smaller volume. This counterintuitive behavior arises because $m = 2/3 < 1$ samples the continuous TM dispersion curve \emph{below} its first integer point, accessing a spectral region entirely inaccessible on geometries constrained to integer $m$. Figure~\ref{fig:wedge_angle} displays this non-monotonic dependence of frequency on cavity opening angle, with the minimum occurring near $\Phi = 270^\circ$ where $m = 2/3$ becomes allowed.

Table~\ref{tab:wedge_spectrum} presents the first six modes of the $90^\circ$ wedge, encompassing both TM and TE polarizations. All modes agree with theory to within $1.4\%$, with mean erro $0.8\%$. The spectrum includes sectoral modes ($k = 0$, hence $\nu = m$) at $m = 2/3, 4/3, 2, \ldots$ and tesseral modes ($k \geq 1$, hence $\nu > m$) such as Mode~3 with $\nu = 5/3$, $m = 2/3$ (where $k = 1$). The tesseral modes are unambiguously identified by their frequencies: no sectoral mode exists between the TM $m = 4/3$ mode at $9.93$~GHz and the TM $m = 2$ mode at $12.32$~GHz, yet Mode~3 appears at $11.25$~GHz, precisely matching the tesseral prediction for $\nu = 5/3$. Similarly, Mode~6 at $13.62$~GHz matches the tesseral prediction for $\nu = 7/3$, $m = 4/3$.

\begin{remark}[Tesseral modes for non-integer $m$]
The existence of tesseral modes for non-integer $m$ follows from the hypergeometric termination condition~\eqref{eq:nu_spectrum_wedge}. When $\nu = m + k$ with $k$ a positive integer, the hypergeometric function in~\eqref{eq:theta_hypergeometric} reduces to a polynomial of degree $k$, which is bounded at $\theta = \pi$. For example, Mode~3 has $m = 2/3$ and $\nu = 5/3$ (so $k = 1$), giving
\begin{equation}
\Theta(\theta) = \sin^{2/3}\theta \cdot {}_2F_1\left(-1, \tfrac{10}{3}; \tfrac{5}{3}; \sin^2(\theta/2)\right) = \sin^{2/3}\theta \cdot (1 - 2\sin^2(\theta/2)) = \sin^{2/3}\theta \cdot \cos\theta,
\end{equation}
which vanishes at both poles and has one nodal plane at $\theta = \pi/2$. This angular structure distinguishes tesseral modes from sectoral modes, which have no nodes between the poles.
\end{remark}

\begin{remark}[Comparison with all-integer case]
For wedge geometries yielding integer $m$ (e.g., the hemisphere with $\Phi = 180^\circ$, giving $m = n$), the condition $\nu = m + k$ reduces to the standard requirement that $\nu$ be an integer $\geq m$. The eigenvalue condition~\eqref{eq:nu_spectrum_wedge} thus unifies the integer and non-integer cases: regularity at both poles requires $\nu - m \in \mathbb{Z}_{\geq 0}$, regardless of whether $m$ itself is an integer.
\end{remark}
\subsection{Conical Truncation: Accessing the Continuous Branch Below $\nu = 1$}
\label{subsec:cone_validation}
While wedge geometries control the azimuthal index $m$, polar conical truncations provide independent control over the angular eigenvalue $\nu$. The physical mechanism is fundamentally different: the cone removes the coordinate singularity at $\theta = 0$, thereby relaxing the regularity constraint that forces $\nu$ to integer values on the full sphere.

On the full sphere, the angular eigenfunctions must remain finite at both poles. Near $\theta = 0$, the associated Legendre equation admits two linearly independent solutions: $P_\nu^m(\cos\theta)$, which is regular, and $Q_\nu^m(\cos\theta)$, which diverges. Demanding regularity at $\theta = 0$ selects the $P_\nu^m$ solution, and the further requirement of regularity at $\theta = \pi$ quantizes $\nu$ to the non-negative integers $\ell = 0, 1, 2, \ldots$ with $\nu = \ell \geq |m|$. This is the origin of the familiar integer angular momentum quantum numbers.

A cone at $\theta = \theta_c$ alters this picture. The angular domain becomes $\theta \in [\theta_c, \pi]$, and the north-pole regularity condition is replaced by a PEC boundary condition at $\theta = \theta_c$. For TM modes, this requires
\begin{equation}
\label{eq:cone_bc}
P_\nu^m(\cos\theta_c) = 0,
\end{equation}
which determines $\nu$ as a function of $\theta_c$ and $m$. Crucially, this boundary condition admits solutions with $\nu < 1$---the \emph{continuous branch} that is entirely absent from the full-sphere spectrum. For TM modes with $m = 0$, the regularity at $\theta = \pi$ (south pole) is automatically satisfied by $P_\nu(\cos\theta)$, so the eigenvalue $\nu$ is determined solely by Eq.~\eqref{eq:cone_bc}.

We solve this eigenvalue problem numerically as follows. For a given cone angle $\theta_c$, we seek the smallest $\nu > 0$ satisfying $P_\nu(\cos\theta_c) = 0$. The Legendre function $P_\nu(x)$ for non-integer $\nu$ is computed via its integral representation or hypergeometric series, and the root is found by standard bracketing methods. Once $\nu(\theta_c)$ is determined, the TM resonant frequency follows from the radial quantization condition $[x j_\nu(x)]'|_{x=ka} = 0$, giving
\begin{equation}
f = \frac{c \, x_{\nu,1}}{2\pi a},
\end{equation}
where $x_{\nu,1}$ is the first positive root of the derivative of the spherical Bessel function.

Figure~\ref{fig:cone_sweep} presents the comparison between this theoretical prediction and HFSS simulations for a north-pole cone with $r_{\mathrm{out}}$ varying from $0.1$ to $10~\mathrm{mm}$ (corresponding to $\theta_c$ from $0.4^\circ$ to $33.7^\circ$). The right axis displays the angular eigenvalue $\nu(\theta_c)$, which increases approximately linearly from $\nu \approx 0.07$ at small cone angles to $\nu \approx 0.41$ at $\theta_c = 33.7^\circ$. This approximately linear dependence,
\begin{equation}
\nu(\theta_c) \approx 0.074 + 0.010 \, \theta_c \quad (\theta_c~\text{in degrees}),
\end{equation}
emerges from the asymptotic behavior of the Legendre function zeros and provides a useful approximation for design purposes.

\begin{figure}[htbp]
\centering
\includegraphics[width=\textwidth]{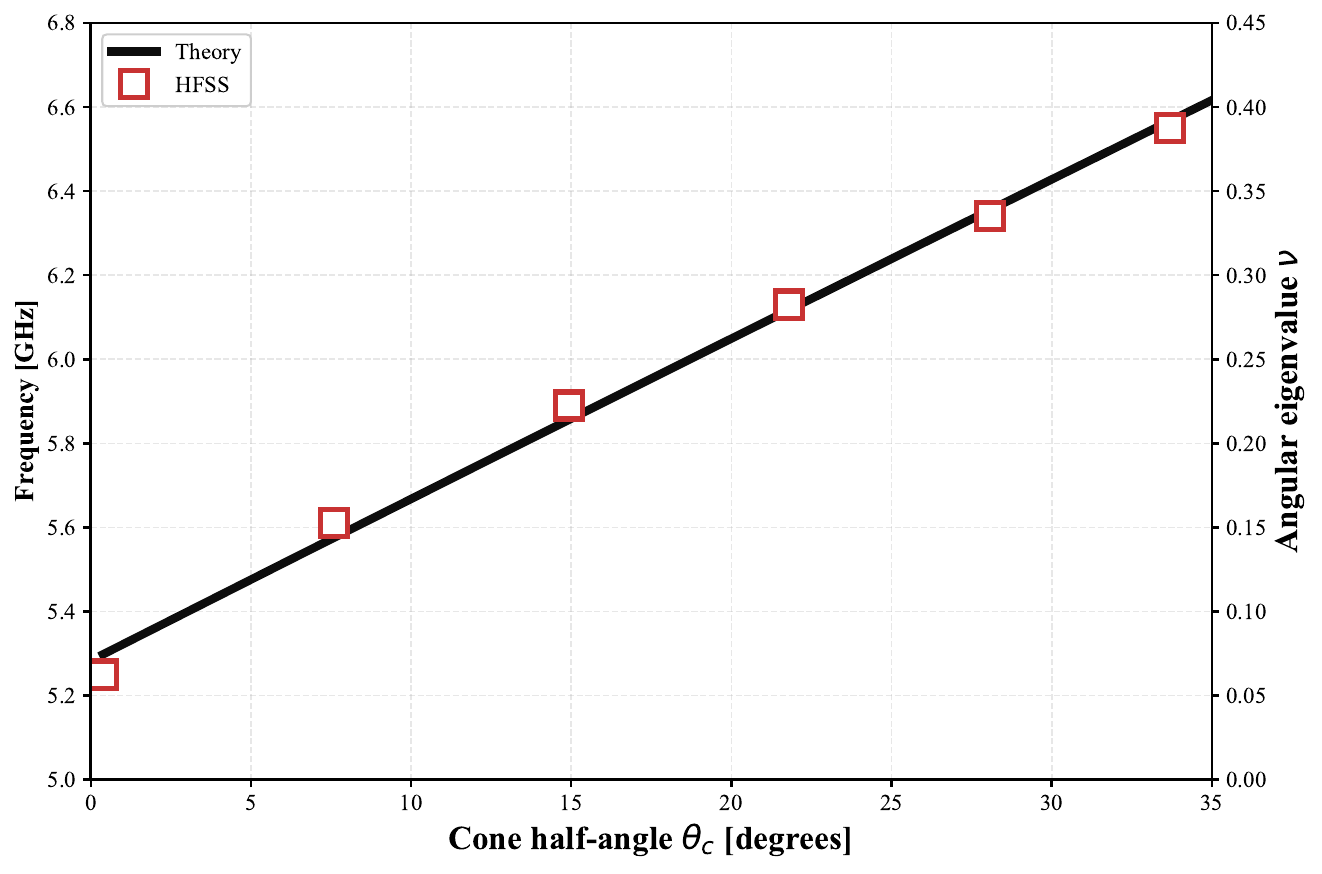}
\caption{Fundamental TM frequency (left axis) and angular eigenvalue $\nu$ (right axis) versus cone half-angle $\theta_c$ for a sphere with a single north-pole cone and no wedge ($m = 0$). The cone removes the $\theta = 0$ regularity constraint, allowing access to the continuous branch with $\nu < 1$. Theory curves are computed by solving
$P_\nu(\cos\theta_c) = 0$ for $\nu$, followed by TM radial quantization; HFSS data (symbols) agree within $0.7\%$. The full sphere ($\theta_c \to 0$) would require $\nu \geq 1$, yielding $f = 8.73~\mathrm{GHz}$—approximately $40\%$ higher than the cone-modified cavity.}

\label{fig:cone_sweep}
\end{figure}

The physical consequence is significant: frequencies range from $5.25~\mathrm{GHz}$ at small cone angles to $6.55~\mathrm{GHz}$ at $\theta_c = 33.7^\circ$, representing a 40\% reduction compared to the full-sphere fundamental at $8.73~\mathrm{GHz}$. This reduction occurs because the cone accesses the $\nu < 1$ branch of the TM dispersion relation, where smaller $\nu$ yields smaller Bessel roots $x_{\nu,1}$ and hence lower frequencies. Table~\ref{tab:cone_spectrum} summarizes the comparison, showing theory--HFSS agreement within $0.7\%$ across the full range of cone angles.

\begin{table}[t]
\centering
\caption{Fundamental TM frequency for sphere with north-pole cone ($a = 15~\mathrm{mm}$, $m = 0$). Angular eigenvalue $\nu$ computed from $P_\nu(\cos\theta_c) = 0$; frequency from TM radial quantization. Agreement is within 0.7\% for all cases.}
\label{tab:cone_spectrum}
\begin{tabular}{cccccc}
\toprule
$r_{\mathrm{out}}$ [mm] & $\theta_c$ [deg] & $\nu$ & $f_{\mathrm{theory}}$ [GHz] & $f_{\mathrm{HFSS}}$ [GHz] & Error [\%] \\
\midrule
0.1 & 0.38 & 0.078 & 5.28 & 5.25 & 0.6 \\
2.0 & 7.59 & 0.150 & 5.58 & 5.61 & $-0.5$ \\
4.0 & 14.93 & 0.224 & 5.86 & 5.89 & $-0.5$ \\
6.0 & 21.80 & 0.292 & 6.11 & 6.13 & $-0.3$ \\
8.0 & 28.07 & 0.354 & 6.35 & 6.34 & 0.2 \\
10.0 & 33.69 & 0.411 & 6.57 & 6.55 & 0.3 \\
\bottomrule
\end{tabular}
\end{table}
The smooth, monotonic frequency increase with cone angle---with no discontinuities or mode crossings---confirms that $\nu$ is indeed a continuous function of the boundary geometry. This behavior cannot be explained by integer quantization: there is no integer $\ell$ between 0 and 1, yet the eigenvalue $\nu$ varies continuously through this interval as the cone angle changes. The cone geometry thus provides direct experimental access to the continuous spectral branch predicted by the non-integer Sturm--Liouville theory.

\subsection{Combined Wedge and Cone: Two-Parameter Spectral Control}
\label{subsec:combined_validation}

The wedge and cone modifications act on different aspects of the angular eigenvalue problem: the wedge quantizes $m$ through azimuthal periodicity, while the cone determines $\nu$ through the polar boundary condition. When both are present, their effects are largely separable, enabling approximate two-parameter control over the cavity spectrum.

For a cavity with opening angle $\Phi$ and north-pole cone at $\theta_c$, the theoretical procedure is: (i) determine $m$ from Eq.~\eqref{eq:m_quantization}, (ii) solve the angular eigenvalue problem on $\theta \in [\theta_c, \pi]$ with boundary condition $P_\nu^m(\cos\theta_c) = 0$ to find $\nu(m; \theta_c)$, and (iii) apply radial quantization for the appropriate polarization. Table~\ref{tab:combined_spectrum} presents HFSS results for three cavity openings combined with a fixed cone ($r_{\mathrm{out}} = 0.1~\mathrm{mm}$, $\theta_c \approx 0.4^\circ$), compared against this theoretical prediction.

\begin{remark}[Relaxation of the $\nu \geq m$ constraint]
Unlike the wedge-only case of Section~\ref{subsec:wedge_validation}, the combined geometry permits $\nu < m$. This is because the cone removes the north pole from the domain, eliminating the regularity constraint that enforced $\nu \geq m$ (equivalently, $k \geq 0$ in Eq.~\eqref{eq:nu_spectrum_wedge}). The eigenvalue $\nu$ is now determined solely by the boundary condition $P_\nu^m(\cos\theta_c) = 0$ at the cone surface and regularity at the south pole, independent of north-pole behavior.
\end{remark}

\begin{table}[htbp]
\centering
\caption{Fundamental TM frequency for combined wedge and cone geometries ($a = 15~\mathrm{mm}$, $r_{\mathrm{out}} = 0.1~\mathrm{mm}$). The wedge determines $m$ via Eq.~\eqref{eq:m_quantization}; the cone modifies $\nu$ through the polar boundary condition. Theory computed by the two-step procedure described in text.}
\label{tab:combined_spectrum}
\begin{tabular}{ccccccc}
\toprule
$\Phi$ [deg] & $m$ & $\nu$ & $f_{\mathrm{theory}}$ [GHz] & $f_{\mathrm{HFSS}}$ [GHz] & Error [\%] \\
\midrule
355 (thin sheet) & 1/2 & 0.08 & 7.02 & 6.94 & 1.2 \\
270 & 2/3 & 0.07 & 7.53 & 7.57 & $-0.5$ \\
180 & 1 & 0.08 & 8.78 & 8.81 & $-0.3$ \\
\bottomrule
\end{tabular}
\end{table}

The results confirm approximate separability of the two parameters. The non-monotonic dependence on cavity opening---with the minimum near $\Phi = 355^\circ$ (thin sheet) arising from $m = 1/2 < 1$---persists in the presence of the cone. The cone produces a nearly uniform frequency shift across all cavity openings, modifying $\nu$ with only weak dependence on $m$ for small $\theta_c$. This approximate separability validates the interpretation of $(\nu, m)$ as nearly independent spectral parameters analogous to $(k_x, k_y)$ in planar waveguides, but here parameterizing the angular degrees of freedom on the modified sphere.

\begin{remark}[Physical interpretation of $\nu < m$]
In Table~\ref{tab:combined_spectrum}, all eigenvalues satisfy $\nu < m$---a regime inaccessible on the full sphere or with wedge-only modifications. Physically, the cone truncates the domain before the field can develop the polar structure required by larger $\nu$. As $\theta_c \to 0$, the eigenvalue $\nu \to m$ from below, recovering the sectoral wedge-only limit. For finite $\theta_c$, the cone ``pushes'' $\nu$ below $m$, enabling access to the spectral region below the sectoral curve.
\end{remark}

\subsection{Summary}
\label{subsec:validation_summary}
The HFSS simulations provide comprehensive validation of the spectral framework across three classes of geometries. 

\textbf{Wedge-only configurations:} Non-integer azimuthal indices $m = 2/3, 4/3, 2, 8/3, \ldots$ (for the $90^\circ$ wedge) are physically realized, with all six tested modes agreeing with theory to within $1.0\%$ (mean error $0.8\%$). Both sectoral modes ($\nu = m$) and tesseral modes ($\nu = m + k$, $k \geq 1$) are observed, with the tesseral modes unambiguously identified by their frequencies falling between consecutive sectoral values. The counterintuitive result that a $90^\circ$ wedge ($\Phi = 270^\circ$) has lower fundamental frequency than a hemisphere ($\Phi = 180^\circ$)---despite larger volume---directly confirms that the dispersion relations are continuous functions of $m$.

\textbf{Cone-only configurations:} The removal of the north-pole singularity grants access to the continuous branch with $\nu < 1$, producing frequencies 40\% below the full-sphere fundamental. The angular eigenvalue increases smoothly with cone angle as predicted by the modified Sturm--Liouville theory, with theory--simulation agreement within $0.7\%$ across the full parameter range.

\textbf{Combined wedge-and-cone configurations:} The two geometric parameters act approximately independently: the wedge controls $m$ while the cone predominantly controls $\nu$, and the resonant frequencies are accurately predicted by the two-step procedure. The cone removes the north-pole regularity constraint, permitting $\nu < m$---a regime inaccessible with wedge-only modifications. Values as low as $m = 1/2$ (thin-sheet geometry) become accessible in this combined configuration. This near-separability enables largely independent tuning of both angular indices, providing a practical pathway to a wide range of $(\nu, m)$ within the accessible parameter space.

Throughout all cases, the theoretical predictions are computed directly from the boundary-value problem---solving Eq.~\eqref{eq:m_quantization} for $m$, Eq.~\eqref{eq:cone_bc} for $\nu$, and the appropriate radial quantization condition for the frequency---with no fitting parameters or empirical adjustments. The consistent sub-percent agreement between theory and simulation establishes that the framework developed in this work accurately describes the electromagnetic phenomenology of domain-modified spherical resonators.

\section{Connection to Antenna Theory}
\label{antenna}
The continuous eigenvalue structure $\nu=\nu(m;\theta_c)$ obtained from truncating the spherical domain by a conical PEC boundary is mathematically closest to the classical separation-of-variables treatment of biconical (and related) antennas developed by Schelkunoff and Papas~\cite{Schelkunoff1943,Papas1950}. In this section we make that connection precise at the level of (i) the \emph{same} angular Sturm--Liouville operator, (ii) the \emph{different} radial radiation/closure conditions, and (iii) the \emph{different} physical interpretation of the singular endpoint behavior. Throughout, we emphasize that ``monopole-like'' or ``biconical-like'' refers to field topology and boundary-value structure, and does not imply any magnetic-monopole source.

\subsection{Common separated structure: the same angular operator}
\label{sec:antenna:common}
For both the closed spherical cavity considered in this paper and the classical biconical antenna region, the separated scalar generator (Debye potential or an equivalent scalar Helmholtz generator) admits the factorization
\begin{equation}
\label{eq:antenna_sep_ansatz}
\psi(r,\theta,\phi)=R(r)\,\Theta(\theta)\,e^{im\phi},
\qquad m\in\mathbb{Z},
\end{equation}
and the polar factor $\Theta(\theta)$ satisfies the associated Legendre Sturm--Liouville equation
\begin{equation}
\label{eq:antenna_theta_SL}
\frac{1}{\sin\theta}\frac{d}{d\theta}\!\left(\sin\theta\,\frac{d\Theta}{d\theta}\right)
+\left[\nu(\nu+1)-\frac{m^2}{\sin^2\theta}\right]\Theta=0,
\qquad \theta\in(\theta_1,\theta_2),
\end{equation}
with eigenvalue parameterized by $\lambda=\nu(\nu+1)$. The defining difference between the \emph{empty full sphere} and \emph{cone-truncated} domains is not the differential operator in \eqref{eq:antenna_theta_SL}, but the admissible endpoint conditions at $\theta=\theta_1,\theta_2$ (regularity vs conductor-enforced conditions). This ``domain is the physics'' viewpoint is exactly what makes $\nu$ discrete (full sphere) or continuous (truncated sphere / conical region).

\subsection{Biconical antennas: radiation closure and complex eigenvalues}
\label{sec:antenna:biconical}
In the Schelkunoff--Papas framework, the physical region is the space \emph{between} two coaxial PEC cones sharing an apex at the origin, with cone half-angle $\theta_c$ (and typically its mirror $\pi-\theta_c$). The PEC condition is applied on the cone surfaces, and the field is required to satisfy an outgoing-wave radiation condition as $r\to\infty$~\cite{Schelkunoff1943,Papas1950}.

At the level of separation of variables, the conical PEC surfaces select $\nu$ through an angular eigencondition of the form
\begin{equation}
\label{eq:antenna_cone_eig_general}
F_{\mathrm{pol}}\!\left(\nu;m,\theta_c\right)=0,
\end{equation}
where the functional $F_{\mathrm{pol}}$ depends on polarization (and on the specific Debye/VSH convention used to define ``TE'' and ``TM'' relative to $\hat r$). In practice, $F_{\mathrm{pol}}$ is obtained by enforcing $\mathrm{n}\times\mathrm{E}=0$ on $\theta=\theta_c$ and $\theta=\pi-\theta_c$ and reducing that field-level condition to an angular boundary condition on $\Theta$ (or on the appropriate Debye scalar). In the symmetric two-cone geometry, the spectrum $\nu=\nu_q(m;\theta_c)$ is continuous in $\theta_c$ and, crucially, becomes \emph{complex} once the radial closure is taken to be radiative.

The radial factor in a radiating antenna is not a standing-wave spherical Bessel function but an outgoing spherical Hankel function. For each admissible $\nu$ determined by \eqref{eq:antenna_cone_eig_general}, the radiating solution takes
\begin{equation}
\label{eq:antenna_radial_outgoing}
R_{\mathrm{rad}}(r)\propto h_{\nu}^{(1)}(kr),
\qquad k=\omega/c,
\end{equation}
so that the far-field behaves as $R_{\mathrm{rad}}(r)\sim e^{ikr}/r$ as $r\to\infty$. Because the radiation condition introduces loss (non-self-adjointness on an unbounded domain), the resulting spectral parameters are generally complex: either $k$ is complex for fixed $\nu$, or equivalently $\nu$ becomes complex for fixed $k$ depending on the chosen formulation. This is the rigorous sense in which the antenna ``eigenvalues'' differ from the cavity eigenvalues: the \emph{same} angular operator \eqref{eq:antenna_theta_SL} appears, but the \emph{radial} condition \eqref{eq:antenna_radial_outgoing} makes the overall problem non-Hermitian.

\subsection{Closed cavity with conical truncation: bounded domain and real spectrum}
\label{sec:antenna:cavity}
Our conical truncation problem is a bounded-domain counterpart of the same separated structure. The conical boundary is part of the cavity wall: in HFSS this corresponds to Boolean-subtracting a cone from the vacuum region so that no metal object remains; the PEC is assigned to the \emph{boundary of the vacuum domain} (sphere plus conical surface). In this setting the radial function is a standing-wave spherical Bessel function,
\begin{equation}
\label{eq:antenna_radial_cavity}
R_{\mathrm{cav}}(r)\propto j_{\nu}(kr),
\end{equation}
and the outer spherical PEC wall at $r=a$ imposes a transcendental quantization condition. In the Debye/VSH convention used in this paper, the final radial conditions take the form
\begin{equation}
\label{eq:antenna_radial_TE}
\text{TE family:}\qquad j_{\nu}(ka)=0,
\end{equation}
and
\begin{equation}
\label{eq:antenna_radial_TM}
\text{TM family:}\qquad \left.\frac{d}{dx}\big[x\,j_{\nu}(x)\big]\right|_{x=ka}=0,
\end{equation}
with $x=ka$ and $k=\omega/c$. Equations \eqref{eq:antenna_radial_TE}--\eqref{eq:antenna_radial_TM} are the transcendental equations referred to in the main text: once $\nu=\nu_q(m;\theta_c)$ is fixed by the conical boundary (single cone or double cone), the roots $x_{q,n}$ of \eqref{eq:antenna_radial_TE} or \eqref{eq:antenna_radial_TM} yield the discrete cavity frequencies
\begin{equation}
\label{eq:antenna_freq_from_root}
f_{q,n}=\frac{c}{2\pi a}\,x_{q,n}.
\end{equation}
Because the domain is bounded and the PEC boundary conditions define a self-adjoint eigenproblem, the spectrum is purely real (up to numerical rounding and any deliberate inclusion of material loss).

\subsection{What is (and is not) the ``antenna connection''?}
\label{sec:antenna:interpretation}
The connection is therefore structural: both problems reduce to the angular operator \eqref{eq:antenna_theta_SL} with $\nu$ selected by conical PEC boundary conditions, and both admit $\nu(\theta_c)$ that varies continuously with cone angle. The essential physical differences are equally clear: antennas are open systems closed by the Sommerfeld radiation condition \eqref{eq:antenna_radial_outgoing} (leading to complex spectra), whereas our cavity is a closed system closed by the spherical PEC wall \eqref{eq:antenna_radial_TE}--\eqref{eq:antenna_radial_TM} (leading to real spectra). This distinction matters when interpreting the small-$\nu$ endpoint behavior. In the cavity, the full-sphere limit enforces pole regularity and recovers the textbook quantization $\nu=\ell\in\mathbb{Z}_{\ge 0}$, with the well-known fact that $(\ell,m)=(0,0)$ yields vanishing electromagnetic fields (no source-free monopole radiation). In contrast, in antenna problems the apex region is physically a feed/defect: the mathematical allowance of singular local behavior near the apex encodes a localized excitation or impedance-supported boundary condition, and the radiating closure shifts the relevant singularities into complex parameter space. This is consistent with the classical observation that Hankel-type solutions are ``waves'' asymptotically and that a radiating structure necessarily includes a reactive near-field region in which the field is not purely propagating~\cite{Schelkunoff1943,Papas1950}.

Finally, we note two quantitative directions where the cavity and antenna viewpoints can be put into direct correspondence, but which we do not fully develop here. First, the characteristic-impedance formulas for biconical/discone structures depend on cone angle through a logarithmic function of $\theta_c$ (e.g. the classical $Z_c(\theta_c)$ expressions in~\cite{Schelkunoff1943,Papas1950}), whereas our cavity eigenvalues encode the same angular geometry through $\nu=\nu_q(m;\theta_c)$ and thus through the root structure of \eqref{eq:antenna_radial_TE}--\eqref{eq:antenna_radial_TM}. Establishing an explicit asymptotic bridge between $Z_c(\theta_c)$ and $\nu_q(m;\theta_c)$ would require matching the apex-regularization and feed model used in the antenna formulation to the corresponding self-adjoint extension (or tip-rounding) used in the cavity formulation. Second, classical electrically-small-antenna limits (Chu--Harrington type bounds) are fundamentally statements about stored energy versus radiated power; our cavity energy analysis provides the natural closed-domain analogue, but converting it into a radiating bound again requires a controlled opening procedure (replacing \eqref{eq:antenna_radial_cavity} by \eqref{eq:antenna_radial_outgoing} and tracking how the eigenvalues move off the real axis).

In summary, the conical-truncation cavity problem can be viewed as a closed-domain analogue of the Schelkunoff--Papas separated-variables description: the \emph{same} angular operator produces a continuous $\nu(\theta_c)$ controlled by geometry, while the \emph{radial} closure distinguishes standing-wave resonances from radiating solutions.

\section{Summary and Conclusions}

This paper establishes several results concerning electromagnetic modes of spherical cavities that extend the classical treatments in standard references. The integer quantization of angular indices $\ell$ and $m$, universally presented as a consequence of regularity and single-valuedness requirements, reflects a choice of boundary conditions rather than an intrinsic constraint of Maxwell's equations. The full sphere with regularity at both poles represents one self-adjoint realization of the angular Laplacian; alternative boundary conditions, implemented through conical or wedge-shaped conducting surfaces, select different self-adjoint extensions with different spectral properties.

The three families of spherical harmonic modes---sectoral, tesseral, and zonal---behave differently under continuation to non-integer indices. The sectoral family, characterized by $\nu = m$, admits a continuous dispersion curve on which $\sin^m\theta$ provides an exact, globally regular solution for any positive real $m$. The tesseral and zonal families exist only at isolated integer points on the full sphere, a restriction traceable to the singular behavior of Legendre functions at the south pole for non-integer degree. Wave impedances for azimuthal power flow satisfy the duality relation $Z_{\mathrm{TE}} Z_{\mathrm{TM}}=-\eta^{2}$ familiar from waveguide theory. We note that the extension to continuous indices permits configurations with $m > \nu$, a regime not accessible with integer spherical harmonics. The physical implications of this extended parameter space merit further investigation.

The dependence of angular eigenvalues on domain geometry has structural parallels in other physical contexts. In general relativity, conical deficits induced by cosmic strings modify azimuthal periodicity analogously to our wedge analysis, leading to non-integer effective angular quantum numbers. The continuous angular branches accessible through polar truncation illustrate how global boundary conditions---rather than local differential equations---determine the admissible mode spectrum. While the present work concerns classical electromagnetic cavities, the underlying Sturm-Liouville structure is shared by wave equations in curved spacetimes, suggesting that analogous geometric control may arise in systems with excised or singular polar regions.

The behavior at the boundary point $(\nu, m) = (0, 0)$ reveals a subtle distinction between potentials and fields. While the Debye potential $\Pi = j_0(kr)$ remains a well-defined, non-singular solution to the scalar Helmholtz equation, the electromagnetic field extracted from it vanishes identically---the potential lies in the kernel of the curl-curl operator that generates physical fields. This situation parallels pure gauge configurations in gauge field theory, where the vector potential is non-zero but the field strength vanishes. The physical interpretation of such ``potential-without-field'' modes, particularly regarding their role in cavity quantization and zero-point energy calculations, remains an open question that connects classical electromagnetic theory to deeper issues in mathematical physics.

Several directions merit further investigation. The three-fold classification of mode families and their distinct behaviors under continuation to non-integer parameters may have a deeper explanation through SO(3) representation theory; whether the continuous sectoral branch admits a natural group-theoretic characterization remains open.

\newpage
\appendix
\section{Connection Formulas for Legendre Functions}
\label{app:connection_formulas}
The behavior of the Legendre function $P_\nu(x)$ near $x = -1$ is determined by the connection formulas relating solutions at the two singular points of the Legendre equation. The fundamental result, derived using standard techniques of hypergeometric function theory~\cite{WhittakerWatson1927, DLMF}, expresses $P_\nu(x)$ in terms of functions centered at $x = -1$.

For $x$ near $-1$, the Legendre function of the first kind admits the representation
\begin{equation}
P_\nu(\cos\theta) = \frac{\sin(\nu\pi)}{\pi}\left[2\ln\left(\sin\frac{\theta}{2}\right) + \psi(\nu+1) + \gamma\right]Q^{(0)}(\theta) + Q^{(1)}(\theta),
\label{eq:Pnu_connection}
\end{equation}
where $\gamma \approx 0.5772$ is the Euler-Mascheroni constant, $\psi$ denotes the digamma function, and $Q^{(0)}$, $Q^{(1)}$ are functions regular at $\theta = \pi$. The crucial observation is that the coefficient $\sin(\nu\pi)$ of the logarithmic term vanishes if and only if $\nu$ is an integer. For non-integer $\nu$, the logarithmic singularity at $\theta = \pi$ renders $P_\nu(\cos\theta)$ inadmissible as a physical solution on the full sphere.

The associated Legendre function $P_\nu^m(x)$ with $m \neq 0$ exhibits analogous behavior. The connection formula takes the form
\begin{equation}
P_\nu^m(\cos\theta) = \frac{\Gamma(\nu+m+1)}{\Gamma(\nu-m+1)}\frac{\sin[(\nu-m)\pi]}{\pi}(\pi-\theta)^{-m}R^{(0)}(\theta) + R^{(1)}(\theta),
\label{eq:Pnum_connection}
\end{equation}
The singular term proportional to $(\pi - \theta)^{-m}$ is absent when $\nu - m$ is a non-negative integer. For the full sphere where $m$ must be an integer, this requires $\nu$ to be an integer at least as large as $|m|$. For wedge geometries where $m$ can be non-integer, this permits non-integer $\nu = m + k$ for non-negative integer $k$, consistent with the eigenvalue condition~\eqref{eq:nu_spectrum_wedge}.These connection formulas establish rigorously the results stated in Theorem~\ref{thm:spectral_structure}: the associated Legendre function $P_\nu^m(\cos\theta)$ is regular at both poles if and only if $\nu$ is a non-negative integer satisfying $\nu \geq |m|$.

\section{Operator-Theoretic Structure of the Angular Problem}
\label{app:operator}
The angular equation \eqref{eq:angular_ode} takes the Sturm-Liouville form
\begin{equation}
-\frac{1}{\sin\theta}\frac{d}{d\theta}\left(\sin\theta\frac{d\Theta}{d\theta}\right) 
+ \frac{m^2}{\sin^2\theta}\Theta = \nu(\nu+1)\Theta
\end{equation}
on the weighted Hilbert space $L^2((0,\pi), \sin\theta\,d\theta)$. Both endpoints 
$\theta = 0$ and $\theta = \pi$ are regular singular points. The Frobenius 
analysis of Section~\ref{math} establishes that solutions behave as 
$\Theta \sim \theta^{\pm|m|}$ near $\theta = 0$ and similarly near $\theta = \pi$.

For $m \neq 0$, only the $\theta^{+|m|}$ branch is square-integrable under the 
weight $\sin\theta \sim \theta$ near $\theta = 0$; the operator is in the 
\emph{limit-circle} case at both endpoints, and selecting the regular branch 
constitutes a boundary condition that determines the self-adjoint extension. 
On the full sphere, demanding regularity at both poles selects a unique 
self-adjoint realization with discrete spectrum $\nu = \ell \in \mathbb{Z}_{\geq |m|}$.
When the domain is truncated to $\theta \in [\theta_c, \pi]$ with a PEC boundary 
at $\theta_c$, the north-pole regularity condition is replaced by 
$\Theta(\theta_c) = 0$ (TM) or $\Theta'(\theta_c) = 0$ (TE). This changes the 
operator domain and hence the spectrum: the eigenvalue $\nu$ now depends 
continuously on $\theta_c$, with integer values recovered only in the limit 
$\theta_c \to 0$.

\section{The Second Solution for the Sectoral Case}
\label{app:second_solution}
Given that $\Theta_1(\theta) = \sin^m\theta$ solves the associated Legendre equation when $\nu = m$, the second linearly independent solution may be constructed by the method of reduction of order. This classical technique, exposited in standard texts on differential equations~\cite{CoddingtonLevinson1955, Arfken2012}, proceeds as follows.

Seeking a second solution in the form $\Theta_2 = \Theta_1 \cdot u(\theta)$ and substituting into the differential equation, one finds after simplification that $u'$ satisfies a first-order equation whose solution is
\begin{equation}
u'(\theta) = \frac{C}{\sin\theta \cdot \Theta_1^2(\theta)} = \frac{C}{\sin\theta \cdot \sin^{2m}\theta} = \frac{C}{\sin^{2m+1}\theta},
\label{eq:u_prime}
\end{equation}
where $C$ is an arbitrary constant. Integration yields
\begin{equation}
u(\theta) = C\int^\theta \frac{d\theta'}{\sin^{2m+1}\theta'}.
\label{eq:u_integral}
\end{equation}
Near $\theta = 0$, the integrand behaves as $\theta'^{-(2m+1)}$, giving
\begin{equation}
u(\theta) \sim -\frac{C}{2m}\theta^{-2m} \quad \text{as } \theta \to 0^+.
\label{eq:u_near_0}
\end{equation}
The second solution therefore satisfies
\begin{equation}
\Theta_2(\theta) = \Theta_1(\theta) \cdot u(\theta) \sim \sin^m\theta \cdot \theta^{-2m} \sim \theta^{-m}
\label{eq:Theta2_behavior}
\end{equation}
near the north pole, exhibiting precisely the singular Frobenius exponent identified in the text.

An analogous analysis near $\theta = \pi$ shows that $\Theta_2$ diverges as $(\pi-\theta)^{-m}$. The second solution is therefore singular at both poles, confirming that $\sin^m\theta$ is the unique globally regular solution of the associated Legendre equation for $\nu = m$.

\section{Spherical Bessel Functions}
\label{app:bessel}
The spherical Bessel functions arise naturally in problems with spherical symmetry and are related to the ordinary Bessel functions by
\begin{equation}
j_\nu(x) = \sqrt{\frac{\pi}{2x}}J_{\nu+1/2}(x), \qquad y_\nu(x) = \sqrt{\frac{\pi}{2x}}Y_{\nu+1/2}(x).
\label{eq:spherical_bessel_def}
\end{equation}
For integer orders, the spherical Bessel functions of the first kind reduce to elementary expressions involving trigonometric functions divided by powers of the argument. The first few are
\begin{align}
j_0(x) &= \frac{\sin x}{x}, \label{eq:j0}\\
j_1(x) &= \frac{\sin x}{x^2} - \frac{\cos x}{x}, \label{eq:j1}\\
j_2(x) &= \left(\frac{3}{x^2} - 1\right)\frac{\sin x}{x} - \frac{3\cos x}{x^2}. \label{eq:j2}
\end{align}

\section{Asymptotic Expansion of Bessel Zeros}
\label{app:zeros}
For large order $\nu$, the zeros of the spherical Bessel function $j_\nu(x)$ and its derivative admit asymptotic expansions that prove useful for understanding the high-frequency behavior of cavity modes.
The $n$th positive zero $x_{\nu,n}$ of $j_\nu(x)$ satisfies the McMahon expansion~\cite{Abramowitz1964, DLMF}
\begin{equation}
x_{\nu,n} = \nu + a_n\nu^{1/3} + \frac{3a_n^2}{20}\nu^{-1/3} + O(\nu^{-1}),
\label{eq:mcmahon}
\end{equation}
where $a_n$ denotes the magnitude of the $n$th zero of the Airy function $\text{Ai}(-a_n) = 0$. The first few values are $a_1 \approx 2.338$, $a_2 \approx 4.088$, $a_3 \approx 5.521$.
For the first zero specifically:
\begin{equation}
x_{\nu,1} \approx \nu + 1.856\nu^{1/3} + O(\nu^{-1/3}).
\label{eq:first_zero_asymp}
\end{equation}
Inverting this relation yields the asymptotic dispersion formula quoted in the text:
\begin{equation}
\nu(x) \approx x - 1.856 x^{1/3} + O(x^{-1/3}),
\label{eq:nu_of_x}
\end{equation}
or in terms of frequency for a cavity of radius $a$:
\begin{equation}
m(f) \approx \frac{2\pi f a}{c} - 1.856\left(\frac{2\pi f a}{c}\right)^{1/3}.
\label{eq:dispersion_asymp}
\end{equation}
This formula provides an excellent approximation to the sectoral dispersion curve for $m \gtrsim 3$, with errors below one percent for $m > 5$.

\end{document}